\documentclass[11pt]{article}

\usepackage[a4paper,margin=1in]{geometry}
\usepackage{amsmath,amssymb,amsthm,amsfonts,mathtools}
\usepackage{mathrsfs}
\usepackage{enumitem}
\usepackage{hyperref}
\usepackage{bm}
\usepackage{stmaryrd}
\usepackage{tikz-cd}
\usepackage{dsfont}

\hypersetup{
  colorlinks=true,
  linkcolor=blue,
  citecolor=blue,
  urlcolor=blue
}

\newtheorem{theorem}{Theorem}[section]
\newtheorem{proposition}[theorem]{Proposition}

\newtheorem{corollary}[theorem]{Corollary}
\newtheorem{definition}[theorem]{Definition}
\newtheorem{remark}[theorem]{Remark}

\newcommand{\R}{\mathbb{R}}

\newcommand{\Rtimes}{\mathbb{R}^{\times}}
\newcommand{\dd}{\mathrm{d}}

\newcommand{\mc}[1]{\mathcal{#1}}

\newcommand{\Lbundle}{\mathcal{L}}

\newcommand{\epsc}{\varepsilon^{\mathrm c}}
\newcommand{\FL}{\mathbb{F}L}

\title{Contact Tulczyjew Geometry for Continuous and Discrete Dissipative Dynamics on Skew Algebroids}
\author{Leonardo Colombo\thanks{Centro de Automática y Robótica (CSIC-UPM), Carretera de Campo Real, km 0, 200, 28500 Arganda del Rey, Spain. (leonardo.colombo@csic.es).}\, and Manuel de Le\'on \thanks{Instituto de Ciencias Matemáticas (CSIC) and Real Academia de Ciencias, Madrid, Spain (mdeleon@icmat.es). \textbf{Funding:} The authors acknowledge financial support from Grant PID2022-137909-NB-C22 funded by the Spanish Ministry of Science and Innovation. \textbf{Author Contribution declaration:} All authors contributed equally to the conception, development, and writing of the manuscript. All authors read and approved the final version.}}
\date{\empty}

\begin{document}
\maketitle

\begin{abstract}
We develop a contact Tulczyjew formalism for dissipative dynamics on skew
algebroids. Starting from the Tulczyjew morphism of an skew algebroid, we
identify its contact extension in a local line-bundle trivialization. The local
representative is obtained by adding to the ordinary Tulczyjew morphism the
Euler vector field contribution on \(E^*\). This gives an intrinsic explanation
of the contact term appearing in the local contact Tulczyjew morphism.

For a contact generating object, the construction produces an implicit
dissipative dynamics on the contact phase side. In local coordinates, the
matching condition gives the Euler--Lagrange--Herglotz equations on the skew
algebroid. In the hyperregular case, the corresponding contact Hamiltonian
equations are recovered by Legendre transformation.

We also develop the discrete counterpart of the construction. After fixing a
discrete admissibility relation, a discrete contact generating object defines a
discrete contact Tulczyjew relation on the contact phase space. Discrete
Herglotz extremals are obtained by matching consecutive contact momenta, with
the usual conormal interpretation in the constrained case. In the regular
tangent-bundle case, this recovers standard contact variational integrators,
while in the singular or skew algebroid setting the same construction remains
meaningful as an implicit discrete relation rather than an a priori update map. 
\end{abstract}

\section{Introduction}

The classical Tulczyjew triple is one of the most effective geometric mechanisms
for passing from generating objects to dynamics \cite{Tulczyjew1976}. In the
symplectic picture, the right-hand side of the triple is Lagrangian, the left-hand
side is Hamiltonian, and the dynamics itself is represented by a generally implicit
first-order differential relation on phase space. One of the main strengths
of this viewpoint is that singular systems are included from the beginning and need
not be represented by an a priori vector field. In this paper $E\to M$ denotes
an skew algebroid, that is, a vector bundle equipped with an anchor and a
skew-symmetric bracket satisfying the Leibniz rule, but not necessarily the Jacobi
identity. For mechanics on such algebroids, the Tulczyjew point of view was
developed in \cite{GG2008}, where the algebroid structure is
encoded by a Tulczyjew morphism
$\varepsilon_E:T^*E\to TE^{*}$.

Contact versions of this picture have been developed from several complementary
viewpoints. In the classical tangent-bundle setting, Esen, Lainz,
de León and Marrero formulated a contact Tulczyjew triple in which contact
Hamiltonian and Herglotz dynamics are represented by Legendrian submanifolds,
using special contact structures and Morse families \cite{EsenLainzDeLeonMarrero2021}.
More recently, Grabowska and Grabowski developed an intrinsic contact Tulczyjew
formalism by replacing cotangent symplectic geometry with the contact geometry of
first jets of sections of line bundles \cite{GG2024}. In that framework, the basic
contact object is not a globally chosen contact form, but the corresponding
line-bundle or homogeneous \(\Rtimes\)-principal-bundle structure. Consequently,
contact Hamiltonians and contact Lagrangians should not be regarded primarily as
ordinary functions, but as local representatives of intrinsic line-bundle objects.
This line of thought is closely related to the modern geometric formulation of
contact Hamiltonian mechanics and Herglotz-type variational principles; see, for
instance, \cite{BravettiCruzTapias2017,deLeonLainz2019}.

The purpose of the present paper is to formulate the corresponding contact
Tulczyjew picture for skew algebroids. We do not view the construction as a
rewriting of the Herglotz equations in different coordinates. The geometric
problem is to understand how the Tulczyjew morphism
$\varepsilon_E:T^*E\to TE^{*}$ of an skew algebroid admits a contact extension such that dissipative dynamics
arises from contact generating objects in the same relation-based way that
ordinary algebroid dynamics arises in the symplectic Tulczyjew picture. Locally,
after choosing a trivialization of the contact line bundle, these generating
objects are represented by contact Lagrangians $L:E\times\mathbb R\to\mathbb R$.
This viewpoint is complementary to the Jacobi, prolongation, and direct
Herglotz variational approaches to dissipative mechanics on Lie algebroids. In
the Jacobi approach, \(E^*\times\mathbb R\) is endowed with a natural Jacobi
structure, and the corresponding contact Hamiltonian equations recover the
Euler-Lagrange-Herglotz equations under regularity assumptions
\cite{ASCLMS2024}. In the prolongation approach, based on the standard prolongation construction for Lie algebroids~\cite{Martinez2001}, one works with the contact Lagrangian and Hamiltonian prolongations $ \mathcal T^E(E\times \mathbb R)\rightarrow E\times \mathbb R$ and $\mathcal T^E(E^*\times \mathbb R)\rightarrow E^*\times \mathbb R$, and the dynamics is described by sections of these contact Lie algebroids \cite{SimoesColomboDeLeonSalgadoSouto2025Annali}. The direct Herglotz
variational approach derives the same local equations from stationarity of the
terminal value of the Herglotz action variable, and also gives
connection-based intrinsic reformulations useful for reduction and
Hamilton--Pontryagin principles
\cite{SimoesColombo2026}. Here, instead, the emphasis is on the
Tulczyjew relation itself: starting from the skew-algebroid Tulczyjew morphism
\(\varepsilon_E:T^*E\to TE^*\), we construct its contact line-bundle extension
and identify the resulting contact term as the Euler contribution on \(E^*\).
Thus the dynamics is generated first as a relation on the contact phase side,
rather than as a vector field or as a section of a prolongation algebroid.

The key observation is simple. In a local trivialization, the contact version of
the algebroid Tulczyjew morphism is obtained from the ordinary morphism \(\varepsilon_E\) by adding the Euler vector field $\Delta_{E^*}$ on \(E^*\). More precisely, if
\(\nu:T^*E\to E^*\) denotes the vertical dual projection and \(r\) is the local
representative of the vertical contact momentum, then the extra contribution is $r\,\Delta_{E^*}(\nu(\alpha_e))\in T_{\nu(\alpha_e)}E^*$, with \(\alpha_e\in T^*_eE\).
In coordinates this is exactly the term \(p_zp_\alpha\) in the momentum equation.
Thus the local contact Tulczyjew morphism is not an ad hoc modification of the
ordinary algebroid morphism, but the local representative of a line-bundle
construction.

Once the contact Tulczyjew morphism has been fixed, the rest of the construction
follows the usual Tulczyjew philosophy. A local Lagrangian
\(L:E\times \mathbb R\to\mathbb R\), representing an intrinsic contact generating
object in a chosen trivialization, determines a contact differential. Composing it
with the contact Tulczyjew morphism gives a generated phase relation. The dynamics
is obtained by requiring the tangent prolongation of the associated Legendre curve
to lie in this relation. In local coordinates, this matching condition gives the
Euler-Lagrange-Herglotz equations on the skew algebroid. The variational
interpretation is recovered by using the admissible variations determined by the
underlying skew algebroid structure.

The relation-based formulation also makes clear the distinction between regular,
hyperregular, and singular dissipative systems. If the fiber Hessian is
nondegenerate, the contact Legendre map is locally invertible and the implicit
relation determines a local vector field. If the contact Legendre map is globally
invertible, one obtains the corresponding contact Hamiltonian description by
Legendre transformation. If the Legendre map is singular, the same equations still
define a meaningful implicit differential relation, but not necessarily a vector
field. This is the setting in which one expects a contact constraint algorithm,
as in the precontact approach to singular Lagrangians
\cite{deLeonLainz2019}; but we do not develop that algorithm here. This relation-based viewpoint also suggests a future Hamilton-Jacobi theory
in terms of contact generating functions, parallel to the role played by
generating functions and symplectic relations in the classical Hamilton--Jacobi theory \cite{AbrahamMarsden1978, FerraroDeLeonMarreroMartinVaquero2017}.

The last part of the paper develops the corresponding discrete picture. The
starting point is the contact-groupoid principle that, before imposing
regularity, the natural discrete object is a Legendrian relation rather than a
map. Since a general skew algebroid need not admit an integrating groupoid, we
work locally after fixing a discrete admissibility relation
\(\mc A_d\subset E\times E\). This relation specifies the admissible consecutive
discrete states and is part of the discretization data. A discrete contact
Lagrangian \(L_d:\mc A_d\times\R\to\R\), together with the discrete Herglotz
update, then defines left and right discrete contact Legendre maps and hence a
discrete contact Tulczyjew relation on \(E^*\times\R\). The discrete equations
are obtained by concatenating this relation, equivalently by matching the
outgoing contact momentum of one step with the incoming contact momentum of the
next. When \(E=TQ\) and \(\mc A_d=Q\times Q\), the construction recovers the
standard contact variational integrators \cite{SimoesMartinDeDiegoLainzDeLeon2021}. Its advantage is that the same
relation remains meaningful for discrete admissibility relations on skew
algebroids, for constrained discrete systems, and for singular cases in which no
discrete phase-space map is available.

This should be distinguished from existing geometric integrators for contact and
Jacobi dynamics. Contact variational integrators start from the discrete Herglotz
principle on \(Q\times Q\times\R\) and, under regularity assumptions, produce
conformal contactomorphisms and exact discrete Lagrangians
\cite{VermeerenBravettiSeri2019,SimoesMartinDeDiegoLainzDeLeon2021}; the general
background is the variational-integrator theory of \cite{MarsdenWest2001}.
Poissonization-based Jacobi Hamiltonian integrators \cite{AraujoOliveiraMestre2026}, on the other hand, lift a
Jacobi Hamiltonian system to a homogeneous Poisson system and construct
Jacobi-preserving time-step maps using homogeneous symplectic bi-realizations and
homogeneous Lagrangian bisections. Our aim is different: we isolate the
Tulczyjew relation behind the discrete Herglotz equations and extend this
relation to the skew algebroid setting. Thus admissibility, Legendre
transformations, contact dissipation, and singularity are encoded in a single
geometric relation, rather than in an a priori discrete map.


The paper is organized as follows. Section~\ref{sec:almost} recalls the skew algebroid
structure, the Tulczyjew morphism, admissible curves and variations, and the local
contact model. Section~\ref{sec:dynamics} derives the implicit contact dynamics, its variational
interpretation, and the contact Legendre correspondence. Section~\ref{sec:intrinsic} explains the
line-bundle origin of the contact Tulczyjew morphism and shows that its local
representative is the morphism used in the preceding sections. Section~\ref{sec:discrete} develops the discrete theory. Starting from the contact-groupoid
motivation, it constructs the local discrete contact Tulczyjew relation
associated with a discrete admissibility relation \(\mc A_d\subset E\times E\).
It then derives the discrete Herglotz equations as momentum matching, relates
the regular tangent-bundle case to contact variational integrators, and explains
why the singular case remains meaningful as an implicit discrete relation. Section~\ref{conc} concludes the paper and outlines several directions for
future work.

\section{Skew algebroids and contact Tulczyjew structures}
\label{sec:almost}

We begin by recalling the minimal algebroid language used in the sequel. Our conventions follow the Tulczyjew approach to mechanics on general algebroids developed in \cite{GG2008}, in which the algebroid structure is encoded by a morphism of double vector bundles $\varepsilon_E: T^*E\rightarrow TE^{*}$ covering the identity on \(E^*\). In the skew-symmetric case this is
equivalently described by an anchor and a skew-symmetric bracket.

\begin{definition}
Let \(\tau:E\to M\) be a vector bundle. An \emph{skew algebroid structure}
on \(E\) consists of a vector bundle morphism $\rho:E\to TM$, called the anchor, and a skew-symmetric bilinear bracket $[\![\cdot,\cdot]\!]:\Gamma(E)\times\Gamma(E)\to\Gamma(E)$ satisfying the Leibniz identity $[\![X,fY]\!]=f[\![X,Y]\!]+(\rho(X)f)Y$ for all \(X,Y\in\Gamma(E)\) and \(f\in C^\infty(M)\). No Jacobi identity
is assumed.
\end{definition}

\begin{remark}
The bracket-and-anchor description will be used for local formulas. Geometrically,
following \cite{GG2008}, the structure is encoded by the corresponding Tulczyjew
morphism \(\varepsilon_E:T^*E\to TE^*\), or equivalently by the associated linear
almost Poisson tensor on \(E^*\).
\end{remark}

\begin{remark}
Our terminology follows the usage of skew algebroids in the sense of
Grabowski \cite{Grabowski2012Modular}. In other parts of the literature the same objects are
also called almost Lie algebroids \cite{GrabowskiJozwikowski2011}, depending on whether
one emphasizes the skew bracket, the failure of the Jacobi identity, or the
associated linear almost Poisson tensor on \(E^*\).
\end{remark}

Let \((x^i,y^\alpha)\) be local vector bundle coordinates on \(E\), and let
\((e_\alpha)\) be the corresponding local basis of sections. The anchor and
the bracket are written as
\[
        \rho(e_\alpha)
        =
        \rho^i_{\alpha}(x)\frac{\partial}{\partial x^i},
        \qquad
        [\![e_\alpha,e_\beta]\!]
        =
        C^\gamma_{\alpha\beta}(x)e_\gamma,
        \qquad
        C^\gamma_{\alpha\beta}=-C^\gamma_{\beta\alpha}.
\]

For \(y=y^\alpha e_\alpha\in E_x\) and
\(\xi=\xi_\alpha e^\alpha\in E_x^*\), we shall use the shorthand
\(C(x)(y,\xi)\in E_x^*\) defined by
$\bigl(C(x)(y,\xi)\bigr)_\alpha
        =
        C^\gamma_{\alpha\beta}(x)y^\beta\xi_\gamma$. With this convention the local expression of the Tulczyjew morphism is
\begin{equation}
\varepsilon_E(x,y,p,\xi)
=
(x,\xi,\rho(x)y,\rho(x)p-C(x)(y,\xi)).
\label{eq:epsE}
\end{equation}
Equivalently,
\[
\dot x^i=\rho^i_\alpha(x)y^\alpha,
\qquad
\dot \xi_\alpha=\rho^i_\alpha(x)p_i-C^\gamma_{\alpha\beta}(x)y^\beta\xi_\gamma.
\]

\begin{definition}
An \emph{admissible curve} \(a(t)=(x^i(t),y^\alpha(t))\) in \(E\) is a
curve satisfying
\begin{equation}
        \dot x^i
        =
        \rho^i_{\alpha}(x)y^\alpha .
\label{eq:admissible}
\end{equation}
\end{definition}

Equivalently, the tangent lift \(\dot a(t)\in T_{a(t)}E\) satisfies $T\tau(\dot a(t))=\rho(a(t))$. Thus admissible curves are those whose tangent lifts take values in the
admissibility subset
\[
T^{\mathrm{adm}}E
:=\{v_e\in TE:\;T\tau(v_e)=\rho(e)\}.
\]
\begin{remark}
In the terminology of \cite{GG2008}, admissible curves
may be characterized intrinsically as those whose tangent prolongations take values in the holonomic subset \(T^{\mathrm{hol}}E\subset TE\). We shall
only need the local expression \eqref{eq:admissible}.
\end{remark}

\begin{definition}\label{var}
Let \(a(t)\) be an admissible curve with base curve \(x(t)\). A
\emph{variation generator} along \(a(t)\) is a time-dependent section $\eta(t)=\eta^\alpha(t)e_\alpha$ along \(x(t)\). The associated infinitesimal admissible variation is given
locally by
\begin{equation}
        \delta x^i
        =
        \rho^i_{\alpha}(x)\eta^\alpha,
        \qquad
        \delta y^\gamma
        =
        \dot\eta^\gamma
        +
        C^\gamma_{\alpha\beta}(x)y^\alpha\eta^\beta .
\label{eq:admvar}
\end{equation}
\end{definition}

\begin{remark}
Formula \eqref{eq:admvar} is the standard local expression of admissible variations
on a general algebroid, following \cite{GG2008}.
It can also be described intrinsically through the relation $\kappa_E$
associated with the Tulczyjew morphism $\varepsilon_E$, but only the
local form \eqref{eq:admvar} will be used below. Thus the admissible
variations are not introduced as an additional choice; they are determined by the algebroid structure itself. For variational problems with fixed
endpoints we shall impose $\eta(t_0)=\eta(t_1)=0$.
\end{remark}

\subsection{Contact Tulczyjew structures and local trivializations}
\label{sec:contact-structure-local}

The classical Tulczyjew picture is formulated in terms of canonical symplectic
structures on iterated tangent and cotangent bundles \cite{Tulczyjew1976}. In the
contact case, the appropriate geometric language is not that of a chosen global
contact form, but that of line bundles and homogeneous
\(\mathbb R^\times\)-principal bundles. Following the contact Tulczyjew
construction of \cite{GG2024}, the contact model is expressed in terms of first
jets of sections of a line bundle, rather than through a globally trivial
contactization \(T^*M\times\mathbb R\).

Let \(E \to M\) be an skew algebroid. Following the line-bundle approach to
contact geometry \cite{GG2024}, the contact structure on the phase side is encoded
by a line bundle $\mathcal L \rightarrow E^*$ or, equivalently, by the associated principal \(\mathbb R^\times\)-bundle $P=\mathcal L^\times\rightarrow E^*$.
After choosing a local trivialization of \(\mathcal L\), the line-bundle coordinate
is represented by an ordinary scalar variable \(z\). Thus the local contact phase
space is represented by \(E^*\times\mathbb R\), with coordinates
\((x^i,p_\alpha,z)\). The coordinate \(z\) is not intrinsic; it depends on the
chosen trivialization.

Similarly, a contact generating object is intrinsically a line-bundle
object. In a local trivialization it is represented by an ordinary smooth
function
\[
        L:E\times\R\rightarrow \R,
        \qquad
        (x,y,z)\mapsto L(x,y,z).
\]

In the same local trivialization, the source of the contact Tulczyjew morphism is represented by coordinates $(x^i,y^\alpha,z,p_i,p_\alpha,p_z,u)$, where \((p_i,p_\alpha)\) are the cotangent variables dual to
\((x^i,y^\alpha)\), \(p_z\) is the contact momentum, and \(u\) is the
contact value variable. The target is represented by coordinates
$(x^i,p_\alpha,z,\dot x^i,\dot p_\alpha,\dot z)$, namely by contactified kinematic variables over \(E^*\times\R\).

The local contact Tulczyjew morphism is
\begin{equation}
\epsc(x,y,z,p_x,p_y,p_z,u)
=
\bigl(
 x,p_y,z,
 \rho(x)y,
 \rho(x)p_x-C(x)(y,p_y)+p_zp_y,
 u
\bigr).
\label{eq:epscmodel}
\end{equation}
Equivalently,
\[
        \dot x^i
        =
        \rho^i_{\alpha}(x)y^\alpha,
        \qquad
        \dot p_\alpha
        =
        \rho^i_{\alpha}(x)p_i
        -
        C^\gamma_{\alpha\beta}(x)y^\beta p_\gamma
        +
        p_zp_\alpha,
        \qquad
        \dot z=u .
\]
\begin{remark}
The term \(p_zp_\alpha\) is the local contact contribution. In
Section~\ref{sec:intrinsic} we show that, in the trivial
contactization, \eqref{eq:epscmodel} is obtained from the ordinary
algebroid Tulczyjew morphism \(\varepsilon_E:T^*E\to TE^*\) by adding the
Euler vector field contribution on \(E^*\).
\end{remark}

The local contact differential of \(L\) is represented by
\[
J^{1,c}L(x,y,z)
=
\left(
 x,y,z,
 \frac{\partial L}{\partial x^i},
 \frac{\partial L}{\partial y^\alpha},
 \frac{\partial L}{\partial z},
 L
\right).
\]
The local contact Legendre map is
\begin{equation}
    \label{CLM}
        \lambda_L(x,y,z)
        =
        \left(
        x,\frac{\partial L}{\partial y},z
        \right).\end{equation}

Applying \(\epsc\) to \(J^{1,c}L\), we obtain the generated phase relation
\[
\Lambda_L(x,y,z)
=
\left(
 x,\frac{\partial L}{\partial y},z,
 \rho(x)y,
 \rho(x)\frac{\partial L}{\partial x}
 -
 C(x)\left(y,\frac{\partial L}{\partial y}\right)
 +
 \frac{\partial L}{\partial z}\frac{\partial L}{\partial y},
 L
\right).
\]

Let $\gamma(t)=(x(t),y(t),z(t))$ be a curve in \(E\times\R\). The local contact Legendre map sends
\(\gamma\) to the phase curve
\[
        \lambda_L\circ\gamma
        =
        \left(
        x(t),
        \frac{\partial L}{\partial y}(x(t),y(t),z(t)),
        z(t)
        \right)
        \in E^*\times\R .
\]
We shall call $\lambda_L\circ\gamma$ the local Legendre curve associated with \(\gamma\).
We denote by $\mathbf t(\lambda_L\circ\gamma)$ its tangent
prolongation, written in local coordinates as
\begin{equation}\label{tangprolong}
 \mathbf t(\lambda_L\circ\gamma)
        =
        \left(
        x,
        \frac{\partial L}{\partial y},
        z,
        \dot x,
        \frac{d}{dt}\frac{\partial L}{\partial y},
        \dot z
        \right).
\end{equation}
The contact Tulczyjew dynamics is defined by the matching condition
\begin{equation}
        \mathbf t(\lambda_L\circ\gamma)
        =
        \Lambda_L\circ\gamma .
\label{eq:kinematicmatch}
\end{equation}
Equating the components in \eqref{eq:kinematicmatch} gives the local
implicit equations derived in the next section.

\begin{remark}
All formulas in this subsection are local representatives of the
\(\mathcal L\)-line-bundle picture. The nontrivial-line-bundle construction
requires the full homogeneous \(\mathbb R^\times\)-principal-bundle formulation. In
Section~\ref{sec:intrinsic} we prove the canonical origin of the above local morphism in the
trivial contactization.
\end{remark}

\section{Implicit dissipative dynamics}
\label{sec:dynamics}

We now derive the local dynamics determined by the contact Tulczyjew
matching condition \eqref{eq:kinematicmatch}.
The guiding point is the same as in the classical Tulczyjew formalism:
the dynamics is first a geometric relation on the phase space, and only
under additional regularity assumptions does it become the graph of a
vector field. Thus singular Lagrangians are not excluded from the
construction.

Let $L:E\times\R\longrightarrow\R$ be the local representative of a contact generating object in a chosen
trivialization of the line bundle. 

\begin{definition}
A curve $\gamma(t)=(x(t),y(t),z(t))\in E\times\R$ is called a \emph{contact Tulczyjew trajectory} for \(L\) if its
Legendre image \(\lambda_L\circ\gamma\) satisfies the matching condition $\mathbf t(\lambda_L\circ\gamma)=\Lambda_L\circ\gamma$.
\end{definition}

\begin{theorem}
\label{thm:implicitdynamics} $\gamma(t)=(x^i(t),y^\alpha(t),z(t))$ is a contact Tulczyjew
trajectory for \(L\) if and only if it satisfies
\begin{align}
\dot x^i
&=
\rho^i_{\alpha}(x)y^\alpha,
\label{eq:ELH1}
\\
\frac{\dd}{\dd t}
\left(
\frac{\partial L}{\partial y^\alpha}
\right)
&=
\rho^i_{\alpha}(x)\frac{\partial L}{\partial x^i}
-
C^\gamma_{\alpha\beta}(x)y^\beta
\frac{\partial L}{\partial y^\gamma}
+
\frac{\partial L}{\partial z}
\frac{\partial L}{\partial y^\alpha},
\label{eq:ELH2}
\\
\dot z
&=
L(x,y,z).
\label{eq:ELH3}
\end{align}
Equivalently,
\begin{equation}
\frac{\dd}{\dd t}
\left(
\frac{\partial L}{\partial y^\alpha}
\right)
+
C^\gamma_{\alpha\beta}(x)y^\beta
\frac{\partial L}{\partial y^\gamma}
-
\rho^i_{\alpha}(x)\frac{\partial L}{\partial x^i}
-
\frac{\partial L}{\partial z}
\frac{\partial L}{\partial y^\alpha}
=
0.
\label{eq:ELHfinal}
\end{equation}
\end{theorem}

\begin{proof}
Let \(\gamma(t)=(x(t),y(t),z(t))\). The Legendre image of \(\gamma\) is
\[
        \lambda_L\circ\gamma
        =
        \left(
        x(t),
        \frac{\partial L}{\partial y}(x(t),y(t),z(t)),
        z(t)
        \right).
\]
Hence its tangent prolongation is locally
\[
        \mathbf t(\lambda_L\circ\gamma)
        =
        \left(
        x,
        \frac{\partial L}{\partial y},
        z,
        \dot x,
        \frac{\dd}{\dd t}\frac{\partial L}{\partial y},
        \dot z
        \right).
\]
On the other hand,
\[
\Lambda_L\circ\gamma
=
\left(
 x,\frac{\partial L}{\partial y},z,
 \rho(x)y,
 \rho(x)\frac{\partial L}{\partial x}
 -
 C(x)\left(y,\frac{\partial L}{\partial y}\right)
 +
 \frac{\partial L}{\partial z}
 \frac{\partial L}{\partial y},
 L
\right)
\]
along the curve. Therefore the matching condition
$\mathbf t(\lambda_L\circ\gamma)
        =
        \Lambda_L\circ\gamma$ is equivalent, component by component, to
\[
        \dot x=\rho(x)y,
\]
\[
        \frac{\dd}{\dd t}
        \frac{\partial L}{\partial y}
        =
        \rho(x)\frac{\partial L}{\partial x}
        -
        C(x)\left(y,\frac{\partial L}{\partial y}\right)
        +
        \frac{\partial L}{\partial z}
        \frac{\partial L}{\partial y},
\]
and
\[
        \dot z=L(x,y,z).
\]
In coordinates these are precisely \eqref{eq:ELH1}-\eqref{eq:ELH3}.
Equivalently, writing the momentum equation with all terms on the left-hand side
gives \eqref{eq:ELHfinal}. Conversely, if these three component equations hold, then the tangent
prolongation of the Legendre curve and the generated phase relation have
the same coordinates along \(\gamma\). Hence
\(t(\lambda_L\circ\gamma)=\Lambda_L\circ\gamma\).
\end{proof}

\begin{definition}
A local representative \(L\) is called \emph{regular} if the fiber Hessian $\left(\frac{\partial^2L}{\partial y^\alpha\partial y^\beta}\right)$ is nondegenerate, and \emph{singular} otherwise.
\end{definition}

Regularity implies local invertibility of the contact Legendre map $\lambda_L:E\times\mathbb R\rightarrow E^*\times\mathbb R$. Indeed, since \(\lambda_L\) is the identity on the variables \((x,z)\), its local
invertibility is equivalent to the nondegeneracy of the fiber Hessian with respect
to \(y\).

If \(L\) is regular, equation (7) can be solved locally for \(\dot y\), because
\(\dot x\) and \(\dot z\) are already determined by (6) and (8). Hence the
implicit relation defined by (6)-(8) determines locally a vector field on
\(E\times\mathbb R\). If \(L\) is singular, equations (6)-(8) define, in general, an implicit
differential relation on \(E\times\mathbb R\), rather than the graph of a vector
field. Thus singularity is not an anomaly of the theory, but part of its natural
domain of applicability.

\subsection{Variational interpretation and Euler-Lagrange-Herglotz equations}
\label{sec:lagrangian}

We now relate the contact-Tulczyjew dynamics to the Herglotz variational
principle on the skew algebroid $E$. The admissibility condition and the
admissible variations are those introduced in Definition~2.3 and
Definition~2.5.

\begin{definition}
Let \(L:E\times\mathbb R\to\mathbb R\) be the local representative of a
contact generating object. A curve $\gamma(t)=(x(t),y(t),z(t))$
is called a Herglotz extremal if \((x(t),y(t))\) is admissible,
$\dot z=L(x,y,z)$, and the terminal value \(z(T)\) is stationary under all admissible
variations with fixed endpoints.
\end{definition}

\begin{theorem}
A curve \(\gamma(t)=(x(t),y(t),z(t))\) is a Herglotz extremal, with respect to
admissible variations with fixed endpoints and fixed initial value of \(z\), if
and only if it is a contact Tulczyjew trajectory.
\end{theorem}

\begin{proof}
Let \(\gamma(t)=(x(t),y(t),z(t))\) be a Herglotz extremal. Since
\(\dot z=L(x,y,z)\), the induced variation of \(z\) satisfies
\[
\frac{d}{dt}(\delta z)
=
\frac{\partial L}{\partial x^i}\delta x^i
+
\frac{\partial L}{\partial y^\alpha}\delta y^\alpha
+
\frac{\partial L}{\partial z}\delta z .
\]
Let
\[
\lambda(t)=
\exp\left(
-\int_0^t
\frac{\partial L}{\partial z}(\tau)\,d\tau
\right).
\]
Since the initial value of \(z\) is fixed, \(\delta z(0)=0\). Multiplying by
\(\lambda(t)\) gives
\[
\frac{d}{dt}\bigl(\lambda(t)\delta z(t)\bigr)
=
\lambda(t)
\left(
\frac{\partial L}{\partial x^i}\delta x^i
+
\frac{\partial L}{\partial y^\alpha}\delta y^\alpha
\right).
\]
Therefore, because \(\lambda(T)\neq 0\), the stationarity condition
\(\delta z(T)=0\) is equivalent to
\[
\int_0^T
\lambda(t)
\left(
\frac{\partial L}{\partial x^i}\delta x^i
+
\frac{\partial L}{\partial y^\alpha}\delta y^\alpha
\right)\,dt=0.
\]

Using the admissible variations of Definition \ref{var}, and integrating by parts, with \(\eta(0)=\eta(T)=0\), we obtain
\[
\int_0^T
\lambda(t)
\left[
\rho^i_\alpha\frac{\partial L}{\partial x^i}
-
C^\gamma_{\alpha\beta}y^\beta
\frac{\partial L}{\partial y^\gamma}
-
\frac{d}{dt}
\left(
\frac{\partial L}{\partial y^\alpha}
\right)
+
\frac{\partial L}{\partial z}
\frac{\partial L}{\partial y^\alpha}
\right]
\eta^\alpha\,dt=0 .
\]
Since the generators \(\eta^\alpha\) are arbitrary, the bracketed term vanishes.
Equivalently,
\[
\frac{d}{dt}
\left(
\frac{\partial L}{\partial y^\alpha}
\right)
+
C^\gamma_{\alpha\beta}y^\beta
\frac{\partial L}{\partial y^\gamma}
-
\rho^i_\alpha
\frac{\partial L}{\partial x^i}
-
\frac{\partial L}{\partial z}
\frac{\partial L}{\partial y^\alpha}
=0 .
\]
Together with the admissibility equation and the Herglotz equation for \(z\), this
is precisely the system characterized in Theorem~3.2. Hence \(\gamma\) is a
contact Tulczyjew trajectory.

Conversely, assume that \(\gamma\) is a contact Tulczyjew trajectory. Then, by
Theorem~3.2, it satisfies the admissibility equation, the Herglotz equation, and
the Euler--Lagrange--Herglotz equations above. Substituting these equations in the
preceding integration-by-parts identity gives
\[
\int_0^T
\lambda(t)
\left(
\frac{\partial L}{\partial x^i}\delta x^i
+
\frac{\partial L}{\partial y^\alpha}\delta y^\alpha
\right)\,dt=0
\]
for every admissible variation with fixed endpoints. Since
\[
\lambda(T)\delta z(T)
=
\int_0^T
\lambda(t)
\left(
\frac{\partial L}{\partial x^i}\delta x^i
+
\frac{\partial L}{\partial y^\alpha}\delta y^\alpha
\right)\,dt
\]
and \(\lambda(T)\neq 0\), it follows that \(\delta z(T)=0\). Thus \(\gamma\) is a
Herglotz extremal.
\end{proof}

\begin{remark}
Thus the local contact-Tulczyjew equations are not merely the coordinate
expression of the matching condition
$t(\lambda_L\circ\gamma)=\Lambda_L\circ\gamma$. They also encode the Herglotz variational principle on the skew algebroid.
In this sense, the Euler-Lagrange-Herglotz equations are the variational
interpretation of the contact Tulczyjew trajectory condition.
\end{remark}

\begin{definition}
The local representative $L$ is called hyperregular if the contact Legendre map
$\lambda_L:E\times\R\to E^*\times\R$ defined in \eqref{CLM} is a global diffeomorphism.
\end{definition}

\begin{proposition}
\label{prop:hyperregular}
Assume that \(L\) is hyperregular. Then there is a unique smooth map $Y:E^*\times\R\to E$, with \(Y(x,p,z)\in E_x\), determined by
$\lambda_L(x,Y(x,p,z),z)=(x,p,z)$, and hence a unique Hamiltonian \(H:E^*\times\R\to\R\) defined by
\[
H(x,p,z)=\langle p,Y(x,p,z)\rangle-L(x,Y(x,p,z),z).
\]
Under the contact Legendre map, the contact-Tulczyjew dynamics generated
by \(L\) is equivalent to the contact-Hamiltonian dynamics generated by
\(H\).
\end{proposition}

\begin{proof}
Since \(L\) is hyperregular, \(\lambda_L\) is a global diffeomorphism. Therefore,
for every \((x,p,z)\in E^*\times\R\), there is a unique
\(y=Y(x,p,z)\in E_x\) such that $p_\alpha=\frac{\partial L}{\partial y^\alpha}(x,y,z)$.
This proves that \(H\) is well defined and smooth.

Differentiating \(H\) with respect to \(p_\alpha\), \(x^i\), and \(z\), and using
\(p_\alpha=\partial L/\partial y^\alpha\), all terms containing derivatives of
\(Y\) cancel. Hence
\[
\frac{\partial H}{\partial p_\alpha}=Y^\alpha,\qquad
\frac{\partial H}{\partial x^i}
=
-\frac{\partial L}{\partial x^i}(x,Y,z),
\qquad
\frac{\partial H}{\partial z}
=
-\frac{\partial L}{\partial z}(x,Y,z).
\]

Let \(\gamma(t)=(x(t),y(t),z(t))\) be a contact Tulczyjew trajectory and set
$p_\alpha(t)=\frac{\partial L}{\partial y^\alpha}(x(t),y(t),z(t))$. By Theorem~\ref{thm:implicitdynamics},
\[
\dot x^i=\rho^i_\alpha(x)y^\alpha,
\]
\[
\dot p_\alpha
=
\rho^i_\alpha(x)\frac{\partial L}{\partial x^i}
-
C^\gamma_{\alpha\beta}(x)y^\beta p_\gamma
+
\frac{\partial L}{\partial z}p_\alpha,
\]
and $\dot z=L(x,y,z)$.

Using the identities above and \(y=Y(x,p,z)\), these equations become
\[
\dot x^i
=
\rho^i_\alpha(x)\frac{\partial H}{\partial p_\alpha},
\]
\[
\dot p_\alpha
=
-\rho^i_\alpha(x)\frac{\partial H}{\partial x^i}
-
C^\gamma_{\alpha\beta}(x)
\frac{\partial H}{\partial p_\beta}p_\gamma
-
\frac{\partial H}{\partial z}p_\alpha,
\]
and
\[
\dot z
=
p_\alpha\frac{\partial H}{\partial p_\alpha}-H.
\]
Thus the Lagrangian and Hamiltonian contact equations are equivalent under
the hyperregular Legendre transformation. Conversely, any solution of the
Hamiltonian equations pulls back by \(\lambda_L^{-1}\) to a solution of the
contact-Tulczyjew equations.
\end{proof}

\begin{remark}
 The singular case should be understood in the relational sense. The generated
object is the subset $\mathcal R_L:=\operatorname{Im}(\Lambda_L)\subset T(E^*\times \mathbb R)$, and the dynamics is defined by the matching condition $t(\lambda_L\circ\gamma)\in \mathcal R_L$. If the contact Legendre map is regular, this condition may be solved locally as
a vector field. If it is hyperregular, this vector field is transformed by the
Legendre map into the contact Hamiltonian dynamics described above. If the
Legendre map is singular, however, no such resolution is available in general,
and the system should be treated as an implicit differential relation \cite{deLeonLainz2019}.

The passage from an implicit differential relation to an actual space of integrable trajectories is a separate geometric problem. In the general theory of implicit differential equations, this is handled by an integrability or constraint algorithm: one projects the relation to the phase space, imposes the tangency condition, projects again, and iterates until an integrable final relation is reached; see the foundational work of Mendella, Marmo and Tulczyjew~\cite{MendellaMarmoTulczyjew1995}. We do not perform here the corresponding stabilization procedure in the contact algebroid setting. Extending such an algorithm to contact Tulczyjew dynamics on skew algebroids would require a precontact constraint theory compatible with algebroid admissibility and with the line-bundle contact structure. This is a separate problem, and is left for future work.
\end{remark}

\begin{remark}
Recall that, for a Lie algebroid \(E\to M\) and a fibration \(P\to M\), the
standard prolongation is
$\mathcal T^E P:=E\times_{TM}TP$,
with anchor given by the projection to \(TP\); see
\cite{Martinez2001}. Thus, in the regular Lie algebroid case, the relation
$\mathcal R_L\subset T(E^*\times\mathbb R)$
can be equivalently described as the image, under the anchor $\rho_\pi:\mathcal T^E(E^*\times\mathbb R)\rightarrow T(E^*\times\mathbb R)$, of the Hamiltonian section of the contact Hamiltonian prolongation
\(\mathcal T^E(E^*\times\mathbb R)\). This is the prolongation-based contact
Hamiltonian picture developed in
\cite{SimoesColomboDeLeonSalgadoSouto2025Annali}, and it is compatible with the
Jacobi description of \(E^*\times\mathbb R\) in
\cite{ASCLMS2024}. The present formulation starts one step earlier: it produces the generated
subset $\mathcal R_L=\operatorname{Im}(\Lambda_L)\subset T(E^*\times\mathbb R)$
directly from the contact Tulczyjew morphism. Therefore it does not require, as
part of the construction, that the relation be the anchor image of a globally
defined or unique dynamical section. This distinction is important for singular
generating objects and for skew algebroids, where we work with the Tulczyjew
morphism and admissibility relation rather than with a full contact Lie
algebroid prolongation.
\end{remark}


\section{The contact Tulczyjew morphism in line-bundle form}
\label{sec:intrinsic}

We now explain the geometric origin of the local contact Tulczyjew morphism used in
Section~2. Let \(E\to M\) be an skew algebroid. The construction is naturally formulated over a line bundle
$\pi_{\Lbundle}:\Lbundle\rightarrow E^*$. After choosing a local trivialization of \(\Lbundle\), the corresponding local
contact phase space is represented by \(E^*\times\R\).

 Let $\varepsilon_E:T^*E\rightarrow TE^{*}$
be the Tulczyjew morphism and we shall also use the
vertical dual projection $\nu:T^*E\rightarrow E^*$, defined by restricting covectors on \(E\) to vertical vectors. In local coordinates,
$\nu(x^i,y^\alpha,p_i,p_\alpha)=(x^i,p_\alpha)$. Let \(\Delta_{E^*}\) denote the Euler vector field of the vector bundle
\(E^*\to M\), locally $\Delta_{E^*}=p_\alpha\frac{\partial}{\partial p_\alpha}$.

The line bundle \(\Lbundle\to E^*\) encodes the contact variable following
the line-bundle approach to contact and Jacobi geometry
\cite{GG2024,BGG2017}. Infinitesimal
motions in \(\Lbundle\) are described
by its Atiyah, or gauge, algebroid
\[
\mathcal D\Lbundle:=T\Lbundle^\times/\Rtimes,
\]
where \(\Lbundle^\times\) is the associated principal \(\Rtimes\)-bundle.
Equivalently, \(\mathcal D\Lbundle\) is the bundle of derivations of
\(\Lbundle\), see, for instance, \cite{Vitagliano2016}. There is a natural projection
$\mathcal D\Lbundle\longrightarrow TE^{*}$
which assigns to each derivation its base vector field. This gives the exact
sequence
\[
0\longrightarrow \operatorname{End}(\Lbundle)
\longrightarrow \mathcal D\Lbundle
\longrightarrow TE^*
\longrightarrow 0.
\]

Since \(\Lbundle\) is a line bundle, every fiber endomorphism of \(\Lbundle\) is multiplication
by a scalar. More precisely, for each \(\xi\in E^*\), an element
\(A\in \operatorname{End}(\Lbundle)_\xi\) is a linear map
\(A:\Lbundle_\xi\to \Lbundle_\xi\). Since \(\Lbundle_\xi\) is one-dimensional, \(A\) is completely
determined by its value on any nonzero element of \(\Lbundle_\xi\). Equivalently, after
choosing a local frame \(s\) of \(\Lbundle\), one has
$A(s(\xi))=a(\xi)s(\xi)$ for a unique scalar \(a(\xi)\). Thus, in the trivialization determined by \(s\),
\(\operatorname{End}(\Lbundle)\) is represented by the trivial line
\(E^*\times\mathbb R\). This identification depends on the chosen frame, whereas
the endomorphism \(A\) itself is intrinsic.

By a local frame of \(\Lbundle\) we mean a nowhere-vanishing local section \(s\) of
\(\Lbundle\). This should not be confused with the local basis \((e_\alpha)\) of
sections of the algebroid \(E\to M\). Such a section determines a local
trivialization of \(\Lbundle\): every element of the fiber \(\Lbundle_\xi\) can be written
uniquely as $\ell=z\,s(\xi)$ for a scalar \(z\). In the same trivialization, the vertical part of a
derivation is represented by a scalar.

Indeed, every derivation \(\delta\in D\Lbundle\) is represented locally by a pair
\((X,a)\), through
\[
\delta(fs)=\bigl(X(f)+af\bigr)s.
\]
Here \(X\) is the induced tangent vector on \(E^*\), while \(a\) is the local
vertical component. The scalar \(a\) depends on the chosen section \(s\), but
the derivation \(\delta\) does not.

In the trivialization \(\Lbundle\simeq E^*\times\mathbb R\), the frame \(s\) plays the
role of the unit section. Thus the scalar \(a\) is equivalently the value of the
vertical endomorphism on this local unit: it is the coefficient defined by $\delta(s)=a\,s$
after the base vector field part has been separated. This is the sense in which the endomorphism component is determined by what it
does to the local unit section of the trivialized line.

The contact contribution is obtained by adding to the ordinary Tulczyjew
direction the Euler direction on \(E^*\). Given \(\alpha_e\in T^*_eE\), the ordinary
Tulczyjew morphism gives
$\varepsilon_E(\alpha_e)\in T_{\nu(\alpha_e)}E^*$.

The Euler vector $\Delta_{E^*}(\nu(\alpha_e))\in T_{\nu(\alpha_e)}E^*$
belongs to the same tangent space. Hence, after choosing a local representative
\(r\) of the vertical contact momentum, the base tangent component of the
contactified morphism is
\[
\varepsilon_E(\alpha_e)+r\,\Delta_{E^*}(\nu(\alpha_e)).
\]

Intrinsically, \(r\) represents an element of
\(\operatorname{End}(\Lbundle)_{\nu(\alpha_e)}\), that is, a fiberwise infinitesimal
rescaling of the contact line over \(\nu(\alpha_e)\). Since the fiber is
one-dimensional, a local frame identifies this endomorphism with multiplication
by a scalar, again denoted by \(r\).

We define the contact momentum bundle over \(T^*E\) by
\[
\mc P^{\mathrm c}(E,\Lbundle)
:=
T^*E\times_{\nu,E^*}\Lbundle
\times_{E^*}\operatorname{End}(\Lbundle)
\times_{E^*}\Lbundle.
\]
A point will be written as \((\alpha_e,\ell,r,u)\), where
$\ell\in\Lbundle_{\nu(\alpha_e)}$, $r\in\operatorname{End}(\Lbundle)_{\nu(\alpha_e)}$ and 
$u\in\Lbundle_{\nu(\alpha_e)}$.
In a local frame, these components are represented by coordinates
$(x,y,p_x,p_y,z,r,u)$. The variable \(z\) represents the line-bundle value \(\ell\), the variable \(r\)
represents the vertical contact momentum, and \(u\) represents the contact value.

A line-bundle contact Tulczyjew morphism is a morphism
\[
\varepsilon_E^{\mathrm c}:
\mc P^{\mathrm c}(E,\Lbundle)
\longrightarrow
\mathcal D\Lbundle\times_{E^*}\Lbundle
\]
such that the projection of its \(\mathcal D\Lbundle\)-component onto \(TE^*\) is
$\varepsilon_E(\alpha_e)+r\,\Delta_{E^*}(\nu(\alpha_e))$. The last factor \(\Lbundle\) records the contact value \(u\). In a local
trivialization this becomes the scalar equation \(\dot z=u\).

\begin{theorem}
\label{thm:global-local-contact}
Let \(s\) be a nowhere-vanishing local section of \(\Lbundle\). In the local trivialization determined
by \(s\), the line-bundle contact Tulczyjew morphism
\(\varepsilon_E^{\mathrm c}\) is represented by
\begin{equation}
\varepsilon_E^{\mathrm c}(x,y,z,p_x,p_y,r,u)
=
\left(
x,p_y,z,
\rho(x)y,
\rho(x)p_x-C(x)(y,p_y)+r\,p_y,
u
\right).
\label{eq:global-local-epsc}
\end{equation}
Equivalently,
\[
\dot x^i=\rho^i_\alpha(x)y^\alpha,\qquad
\dot p_\alpha
=
\rho^i_\alpha(x)p_i
-
C^\gamma_{\alpha\beta}(x)y^\beta p_\gamma
+
r p_\alpha,\qquad
\dot z=u.
\]
After the identification \(r=p_z\), this is the local contact Tulczyjew morphism
of Section~2.
\end{theorem}

\begin{proof}
In the local frame \(s\), the ordinary Tulczyjew morphism is the map recalled in
Section~2,
\[
\varepsilon_E(x,y,p_x,p_y)
=
\bigl(x,p_y,\rho(x)y,\rho(x)p_x-C(x)(y,p_y)\bigr).
\]

The Euler vector field on \(E^*\) is locally
$\Delta_{E^*}
=
p_\alpha\frac{\partial}{\partial p_\alpha}$. Therefore, at the point \((x,p_y)\in E^*\), its contribution is represented by
$r\,\Delta_{E^*}(x,p_y)=r\,p_\alpha\frac{\partial}{\partial p_\alpha}$. Thus the tangent component on \(E^*\) of the contactified morphism is represented by
\[
\bigl(
x,p_y,
\rho(x)y,
\rho(x)p_x-C(x)(y,p_y)+r\,p_y
\bigr).
\]
The element \(\ell\in\Lbundle_{\nu(\alpha_e)}\) is represented in the frame \(s\)
by the scalar coordinate \(z\), and the contact value
component \(u\) is represented by the scalar equation \(\dot z=u\). This gives
\eqref{eq:global-local-epsc}.
\end{proof}

\begin{corollary}
\label{cor:trivial-contactization}
If \(\Lbundle=E^*\times\R\) is the trivial line bundle and the standard frame is
chosen, then
\[
\mc P^{\mathrm c}(E,\Lbundle)
\simeq
T^*E\times\R_z\times\R_{p_z}\times\R_u,
\]
and \(\varepsilon_E^{\mathrm c}\) is exactly the morphism
\[
\varepsilon_E^{\mathrm c}(x,y,z,p_x,p_y,p_z,u)
=
\left(
x,p_y,z,
\rho(x)y,
\rho(x)p_x-C(x)(y,p_y)+p_zp_y,
u
\right)
\]
used in the local formulation.
\end{corollary}

\begin{proof}
For the trivial line bundle, the standard frame is global. Hence the line-bundle
coordinate, the vertical contact momentum, and the contact value are represented
globally by the scalars \(z\), \(p_z\), and \(u\), respectively. In particular,
\(\operatorname{End}(\Lbundle)\simeq E^*\times\R\), and the element
\(r\in\operatorname{End}(\Lbundle)\) is represented globally by \(p_z\).
Therefore the local representative of Theorem~\ref{thm:global-local-contact} is
global, and relabelling \(r=p_z\) gives the displayed formula.
\end{proof}

We now consider generating objects. Intrinsically, a contact generating object is a
line-bundle object. After choosing a local frame of \(\Lbundle\), it is represented
by a smooth function \(L:E\times\R\to\R\). For each fixed \(z\), write
\(L_z:E\to\R\) for \(L_z(e)=L(e,z)\). The local representative of its contact
differential is
\[
J^{1,\mathrm c}L(e,z)
=
\bigl(\dd_eL_z,z,\partial_zL(e,z),L(e,z)\bigr).
\]
Thus, in coordinates,
\[
J^{1,\mathrm c}L(x,y,z)
=
\left(
x,y,z,
\frac{\partial L}{\partial x},
\frac{\partial L}{\partial y},
\frac{\partial L}{\partial z},
L
\right),
\]
with the same ordering convention as in Section~2.

Composing this local contact differential with the line-bundle contact Tulczyjew morphism gives
the generated phase relation
\[
\Lambda_L:=\varepsilon_E^{\mathrm c}\circ J^{1,\mathrm c}L.
\]

\begin{proposition}
\label{prop:local-generated-relation}
In a local trivialization of \(\Lbundle\), the generated phase relation is
\[
\Lambda_L(x,y,z)
=
\left(
x,
\frac{\partial L}{\partial y},
z,
\rho(x)y,
\rho(x)\frac{\partial L}{\partial x}
-
C(x)\left(y,\frac{\partial L}{\partial y}\right)
+
\frac{\partial L}{\partial z}
\frac{\partial L}{\partial y},
L
\right).
\]
Equivalently,
\[
\dot x^i=\rho^i_\alpha(x)y^\alpha,\qquad
\dot p_\alpha=
\rho^i_\alpha(x)\frac{\partial L}{\partial x^i}
-
C^\gamma_{\alpha\beta}(x)y^\beta
\frac{\partial L}{\partial y^\gamma}
+
\frac{\partial L}{\partial z}
\frac{\partial L}{\partial y^\alpha},
\qquad
\dot z=L.
\]
\end{proposition}

\begin{proof}
Substituting the local representative of the contact differential into
\eqref{eq:global-local-epsc} gives
\[
p_y=\frac{\partial L}{\partial y},\qquad
p_x=\frac{\partial L}{\partial x},\qquad
r=\frac{\partial L}{\partial z},\qquad
u=L.
\]
The formula follows immediately.
\end{proof}

The contact Legendre map is obtained by taking the vertical part of the contact differential.
Globally this is the map over \(E^*\) induced by the vertical projection
\(\nu:T^*E\to E^*\). In a local frame it is represented by \eqref{CLM}.

\begin{theorem}
\label{thm:recovery-local-dynamics}
Let \(\gamma(t)=(x(t),y(t),z(t))\) be a local representative of a curve in the
contactified algebroid side. The matching condition
\[
\mathbf t(\lambda_L\circ\gamma)=\Lambda_L\circ\gamma
\]
is represented, in any local trivialization, by the system
\[
\dot x^i=\rho^i_\alpha(x)y^\alpha,
\]
\[
\frac{\dd}{\dd t}
\left(
\frac{\partial L}{\partial y^\alpha}
\right)
=
\rho^i_\alpha(x)\frac{\partial L}{\partial x^i}
-
C^\gamma_{\alpha\beta}(x)y^\beta
\frac{\partial L}{\partial y^\gamma}
+
\frac{\partial L}{\partial z}
\frac{\partial L}{\partial y^\alpha},
\]
and
\[
\dot z=L(x,y,z).
\]
Thus the line-bundle construction recovers exactly the local contact Tulczyjew dynamics
of Theorem~\ref{thm:implicitdynamics}.
\end{theorem}

\begin{proof}
The tangent prolongation of the local Legendre curve is given by \eqref{tangprolong},
\[
\mathbf t(\lambda_L\circ\gamma)
=
\left(
x,
\frac{\partial L}{\partial y},
z,
\dot x,
\frac{\dd}{\dd t}\frac{\partial L}{\partial y},
\dot z
\right).
\]
By Proposition~\ref{prop:local-generated-relation}, \(\Lambda_L\circ\gamma\) is represented
by
\[
\left(
x,
\frac{\partial L}{\partial y},
z,
\rho(x)y,
\rho(x)\frac{\partial L}{\partial x}
-
C(x)\left(y,\frac{\partial L}{\partial y}\right)
+
\frac{\partial L}{\partial z}
\frac{\partial L}{\partial y},
L
\right).
\]
Equating components gives precisely the stated system.
\end{proof}

\begin{remark}
The scalar variables \(z\), \(r\), and \(u\) are not intrinsic objects. They are
local representatives of line-bundle data: \(z\) represents the line-bundle value,
\(r\) represents the vertical contact momentum, and \(u\) represents the contact
value which becomes \(\dot z\) in a local trivialization. The local term
\(p_zp_\alpha\) appearing in Section~2 is the representative of the Euler
contribution \(r\,\Delta_{E^*}\) to the tangent component on \(E^*\) of the
line-bundle contact Tulczyjew morphism.
\end{remark}

\section{Discrete contact dynamics from Tulczyjew relations}
\label{sec:discrete}

We now describe the discrete counterpart of the previous construction. The guiding
principle is that a generating object should define a geometric relation, and that
the discrete dynamics should be obtained by matching consecutive momenta. This is
the mechanism underlying discrete Lagrangian mechanics on symplectic groupoids
\cite{MarreroMartinDeDiegoStern2015,ColomboMartindeDiego2016,Colombo2014},
and it will serve here as the model for the contact construction.

In the symplectic case, the model example is the cotangent groupoid. If
\(\mathcal G\rightrightarrows P\) is a Lie groupoid with Lie algebroid
\(A\mathcal G\), then \(T^*\mathcal G\) carries a canonical symplectic groupoid
structure $T^*\mathcal G\rightrightarrows A^*\mathcal G$. Its source and target maps, which we denote by
$\widetilde\alpha,\widetilde\beta:T^*\mathcal G\longrightarrow A^*\mathcal G$, are the cotangent lifts of the infinitesimal source and target structures of
\(\mathcal G\). The precise formulas are not needed here; what matters is that
\(\widetilde\alpha\) and \(\widetilde\beta\) assign to a covector on
\(\mathcal G\) its left and right momenta in \(A^*\mathcal G\).

More generally, if
\(\widetilde{\mathcal G}\rightrightarrows P\) is a symplectic groupoid with
source and target maps
$\widetilde\alpha,\widetilde\beta:\widetilde{\mathcal G}\longrightarrow P$, then a Lagrangian submanifold
\(\Sigma\subset\widetilde{\mathcal G}\) defines a discrete relation by the
source--target matching condition
\[
\widetilde\beta(\mu_k)=\widetilde\alpha(\mu_{k+1}).
\]
For the cotangent groupoid
\(T^*\mathcal G\rightrightarrows A^*\mathcal G\), a discrete Lagrangian
\(L_d:\mathcal G\to\R\) generates the Lagrangian submanifold $\dd L_d(\mathcal G)\subset T^*\mathcal G$,
and this matching condition gives the usual discrete Euler--Lagrange equations
on Lie groupoids \cite{MarreroMartinDeDiegoStern2015}.

The contact analogue replaces Lagrangian submanifolds of symplectic groupoids
by Legendrian submanifolds of contact groupoids.

\begin{definition}
A \emph{contact groupoid} is a Lie groupoid
$\mathcal G^{\mathrm c}\rightrightarrows P^{\mathrm c}$
together with a multiplicative contact structure on \(\mathcal G^{\mathrm c}\).
Equivalently, in the coorientable case, it is represented locally by a contact
form \(\theta\in\Omega^1(\mathcal G^{\mathrm c})\) and a multiplicative function
\(\sigma:\mathcal G^{\mathrm c}\to\mathbb R\) such that, on the manifold of
composable pairs \((\mathcal G^{\mathrm c})^{(2)}\),
\begin{equation}\label{contgro}
m^*\theta
=
\operatorname{pr}_1^*\theta
+
e^{\sigma\circ\operatorname{pr}_1}\operatorname{pr}_2^*\theta,\end{equation} where \(m\) is the groupoid multiplication and
\(\operatorname{pr}_1,\operatorname{pr}_2:(\mathcal G^{\mathrm c})^{(2)}
\to\mathcal G^{\mathrm c}\) are the projections. In the line-bundle language,
this means that the contact distribution is encoded by a contact line bundle
and that the groupoid multiplication is compatible with this line-bundle contact
structure. Thus a contact groupoid is the contact counterpart of a symplectic
groupoid, where the symplectic form is replaced by a multiplicative contact
structure; see, for instance,
\cite{CrainicZhu2007,IglesiasMarrero2003,BGG2017}.
\end{definition}

\begin{remark}
There are several equivalent, or partially equivalent, ways of presenting
contact groupoids in the literature. In Definition~5.1 we use the local
coorientable form of the notion, following the convention of
\cite{IglesiasMarrero2003}: a contact groupoid is represented by a contact form
\(\theta\) and a multiplicative function \(\sigma\) satisfying the conformal
multiplicativity condition \eqref{contgro}. This is the most convenient representative for the local discussion below,
where the contact line bundle is trivialized and the contact variable is written
as an ordinary scalar.

This should be distinguished from the more intrinsic formulation of contact
groupoids in terms of homogeneous symplectic principal
\(\mathbb R^\times\)-bundles, or equivalently \(\mathbb R^\times\)-groupoids
with a homogeneous multiplicative symplectic structure, as emphasized in
\cite{BGG2017}. That formulation does not require the choice of a global contact
form and is better adapted to nontrivial contact line bundles. In the
coorientable/trivialized case, however, the homogeneous formulation reduces
locally to the conformal contact-form description above. Thus, throughout this
section, the contact-form definition should be read as a local representative of
the intrinsic line-bundle picture, rather than as a restriction on the global
contact geometry.
\end{remark}

\begin{definition}
Let \((C,\mathcal H)\) be a contact manifold, where
\(\mathcal H\subset TC\) is the contact distribution. A submanifold
\(\Sigma\subset C\) is called \emph{Legendrian} if
$T\Sigma\subset \mathcal H$ and \(\Sigma\) is maximal with this property. Equivalently, if the contact
structure is locally represented by a contact form \(\theta\), then $\theta|_{\Sigma}=0$ and
\[
\dim \Sigma=\frac{1}{2}(\dim C-1).
\]
In the line-bundle formulation, this condition is independent of the chosen
local contact form and means that \(\Sigma\) is a maximal submanifold tangent to
the intrinsic contact distribution.
\end{definition}

Thus, if $\mathcal G^{\mathrm c}\rightrightarrows P^{\mathrm c}$ is a contact groupoid, a Legendrian submanifold $\Sigma^{\mathrm c}\subset \mathcal G^{\mathrm c}$ is the contact-geometric counterpart of a Lagrangian submanifold in a symplectic
groupoid. This is the groupoid-level object which generates a contact relation
on the base. In this sense, Legendrian submanifolds play, in contact geometry,
the role played by Lagrangian submanifolds in symplectic geometry; see
\cite{Geiges2008,EsenLainzDeLeonMarrero2021}.

Suppose then that a contact groupoid
$\mathcal G^{\mathrm c}\rightrightarrows P^{\mathrm c}$
is given, with source and target maps
$\alpha^{\mathrm c},\beta^{\mathrm c}:
\mathcal G^{\mathrm c}\longrightarrow P^{\mathrm c}$.
Here \(P^{\mathrm c}\) denotes the contact phase space, or its line-bundle
contactification, and \(\mathcal G^{\mathrm c}\) denotes the groupoid carrying
the multiplicative contact structure.

The pair of source and target maps restricts to a smooth map
$(\alpha^{\mathrm c},\beta^{\mathrm c})|_{\Sigma^{\mathrm c}}:
\Sigma^{\mathrm c}\rightarrow P^{\mathrm c}\times P^{\mathrm c}$. Its image defines the contact relation generated by
\(\Sigma^{\mathrm c}\),
\[
R_{\Sigma^{\mathrm c}}
:=
(\alpha^{\mathrm c},\beta^{\mathrm c})(\Sigma^{\mathrm c})
\subset P^{\mathrm c}\times P^{\mathrm c}.
\]
Equivalently,
\begin{equation}\label{RS}
R_{\Sigma^{\mathrm c}}
=
\left\{
(p_-,p_+)\in P^{\mathrm c}\times P^{\mathrm c}
\;:\;
\exists\,\eta\in\Sigma^{\mathrm c},\quad
p_-=\alpha^{\mathrm c}(\eta),\quad
p_+=\beta^{\mathrm c}(\eta)
\right\}.
\end{equation}
Thus \(R_{\Sigma^{\mathrm c}}\) is a relation on the contact phase space
\(P^{\mathrm c}\), not necessarily the graph of a map. When the restricted map
\((\alpha^{\mathrm c},\beta^{\mathrm c})|_{\Sigma^{\mathrm c}}\) has constant
rank, \(R_{\Sigma^{\mathrm c}}\) is an immersed submanifold of
\(P^{\mathrm c}\times P^{\mathrm c}\); in general, it should be understood as
an immersed relation.

A discrete trajectory is obtained by concatenating generated steps. A two-step
concatenation is possible precisely when the final phase point of the first
step coincides with the initial phase point of the second one. In terms of the
source and target maps, this means that a pair
\((\eta_k,\eta_{k+1})\in\Sigma^{\mathrm c}\times\Sigma^{\mathrm c}\) lies in
the fiber product
\[
\Sigma^{\mathrm c}
\times_{\beta^{\mathrm c},\,P^{\mathrm c},\,\alpha^{\mathrm c}}
\Sigma^{\mathrm c}
:=
\left\{
(\eta_k,\eta_{k+1})\in\Sigma^{\mathrm c}\times\Sigma^{\mathrm c}
\;:\;
\beta^{\mathrm c}(\eta_k)=\alpha^{\mathrm c}(\eta_{k+1})
\right\}.
\]
Therefore a discrete trajectory is a sequence
\((\eta_k)\subset\Sigma^{\mathrm c}\) such that $\beta^{\mathrm c}(\eta_k)=\alpha^{\mathrm c}(\eta_{k+1})$ for all \(k\). If $p_k
:=
\beta^{\mathrm c}(\eta_{k-1})
=
\alpha^{\mathrm c}(\eta_k)$, then each element \(\eta_k\in\Sigma^{\mathrm c}\) determines a step
\[
(p_{k-1},p_k)
=
\left(
\alpha^{\mathrm c}(\eta_k),
\beta^{\mathrm c}(\eta_k)
\right)
\in R_{\Sigma^{\mathrm c}}.
\]
Thus the induced sequence \((p_k)\) in \(P^{\mathrm c}\) satisfies $(p_{k-1},p_k)\in R_{\Sigma^{\mathrm c}}\text{ for all } k$.

Conversely, a sequence \((p_k)\) satisfying this condition is a trajectory of
the relation \(R_{\Sigma^{\mathrm c}}\), whenever the corresponding steps can
be lifted to elements \(\eta_k\in\Sigma^{\mathrm c}\). Hence the source--target
matching condition on \(\Sigma^{\mathrm c}\) is the lifted form of the
trajectory condition for the relation \(R_{\Sigma^{\mathrm c}}\).

\medskip

Before turning to the local algebroid construction, let us also recall the
regular case. A local bisection of a groupoid $\mathcal G^{\mathrm c}\rightrightarrows P^{\mathrm c}$ is a submanifold \(B\subset\mathcal G^{\mathrm c}\) such that the restrictions $\alpha^{\mathrm c}|_B:B\to P^{\mathrm c}$ and $\beta^{\mathrm c}|_B:B\to P^{\mathrm c}$ are local diffeomorphisms onto open subsets of \(P^{\mathrm c}\). Equivalently,
near each point of \(B\), the source map parametrizes the bisection and the
target map defines a local transformation of the base,
$\Phi_B
=
\beta^{\mathrm c}|_B
\circ
\left(\alpha^{\mathrm c}|_B\right)^{-1}$. This is the standard notion of local bisection used in the groupoid formulation
of discrete mechanics \cite{MarreroMartinDeDiegoStern2015} and in Jacobi/contact groupoids
\cite{IglesiasMarrero2003,CrainicZhu2007}.

If \(P^{\mathrm c}\) is represented locally by a contact form \(\theta\), a
local diffeomorphism \(\Phi:P^{\mathrm c}\to P^{\mathrm c}\) is a contact
transformation when it preserves the contact distribution,
\[
T\Phi(\ker\theta)=\ker\theta .
\]
Equivalently, there exists a nowhere-vanishing smooth function \(f\) such that
\(\Phi^*\theta=f\,\theta\). Thus, after choosing a local contact form, contact
transformations are represented by conformal contact transformations
\cite{Geiges2008}. In the line-bundle formulation, this is the local expression
of preserving the intrinsic contact line-bundle structure; see, for instance,
\cite{BGG2017,GG2024}.

For contact groupoids, local Legendrian bisections act on the base by contact
transformations \cite{IglesiasMarrero2003, CrainicZhu2007}. This is the contact analogue of the standard symplectic
groupoid fact that Lagrangian bisections act by Poisson transformations on the
base \cite{MarreroMartinDeDiegoStern2015}. Equivalently, it follows by passing to the associated homogeneous
symplectic \(\Rtimes\)-principal bundle
\cite{BGG2017}. Hence, if the relation
\(R_{\Sigma^{\mathrm c}}\) is locally generated by a Legendrian bisection, it is
locally the graph of a contact, or conformal contact, transformation. Without
this regularity, the natural object remains the relation
\(R_{\Sigma^{\mathrm c}}\).

\begin{theorem}
\label{thm:regular-legendrian-relation}
Let \(\mathcal G^{\mathrm c}\rightrightarrows P^{\mathrm c}\) be a contact
groupoid, and let \(\Sigma^{\mathrm c}\subset\mathcal G^{\mathrm c}\) be a
Legendrian submanifold. Let
$R_{\Sigma^{\mathrm c}}
:=
(\alpha^{\mathrm c},\beta^{\mathrm c})(\Sigma^{\mathrm c})
\subset P^{\mathrm c}\times P^{\mathrm c}$
be the relation \eqref{RS} defined by \(\Sigma^{\mathrm c}\). Assume that, near
\(\eta_0\in\Sigma^{\mathrm c}\), the restriction
$\alpha^{\mathrm c}|_{\Sigma^{\mathrm c}}:\Sigma^{\mathrm c}\rightarrow
P^{\mathrm c}$ is a local diffeomorphism onto its image. Then, near
\((\alpha^{\mathrm c}(\eta_0),\beta^{\mathrm c}(\eta_0))\), the relation
\(R_{\Sigma^{\mathrm c}}\) is the graph of the local map
$\Phi_{\Sigma^{\mathrm c}}
=
\beta^{\mathrm c}|_{\Sigma^{\mathrm c}}
\circ
\left(\alpha^{\mathrm c}|_{\Sigma^{\mathrm c}}\right)^{-1}$.

If, moreover, \(\Sigma^{\mathrm c}\) is a local Legendrian bisection, then
\(\Phi_{\Sigma^{\mathrm c}}\) is a local contact transformation of
\(P^{\mathrm c}\). Equivalently, after choosing a local contact form \(\theta\),
there exists a nowhere-vanishing smooth function \(f\) such that
$\Phi_{\Sigma^{\mathrm c}}^*\theta=f\,\theta$.
\end{theorem}

\begin{proof}
We first prove the graph statement. Since
\(\alpha^{\mathrm c}|_{\Sigma^{\mathrm c}}\) is a local diffeomorphism near
\(\eta_0\), there exist neighborhoods $U\subset\Sigma^{\mathrm c}$, and $V\subset P^{\mathrm c}$, with \(\eta_0\in U\) and \(\alpha^{\mathrm c}(\eta_0)\in V\), such that
$\alpha^{\mathrm c}|_U:U\rightarrow V$ is a diffeomorphism. Therefore every source point \(p_-\in V\) determines a
unique element $\eta
=
\left(\alpha^{\mathrm c}|_U\right)^{-1}(p_-)
\in U$.

The corresponding target point is then uniquely determined by
$p_+
=
\beta^{\mathrm c}(\eta)
=
\beta^{\mathrm c}|_U
\circ
\left(\alpha^{\mathrm c}|_U\right)^{-1}(p_-)$. Hence, locally near \((\alpha^{\mathrm c}(\eta_0),\beta^{\mathrm c}(\eta_0))\), one has
\[
R_{\Sigma^{\mathrm c}}\cap (V\times \beta^{\mathrm c}(U))
=
\{(p^-,\Phi_{\Sigma^{\mathrm c}}(p^-)):p^-\in V\},
\]
where
\[
\Phi_{\Sigma^{\mathrm c}}
=
\beta^{\mathrm c}|_U\circ(\alpha^{\mathrm c}|_U)^{-1}.
\]
Indeed, the local invertibility of \(\alpha^{\mathrm c}|_U\) gives a unique
\(\eta\in U\) over each \(p^-\in V\), and hence a unique corresponding
target point \(p^+=\beta^{\mathrm c}(\eta)\). Thus
\(R_{\Sigma^{\mathrm c}}\) is locally the graph of
\(\Phi_{\Sigma^{\mathrm c}}\).

Assume now that \(\Sigma^{\mathrm c}\) is a local bisection. Then both
\(\alpha^{\mathrm c}|_{\Sigma^{\mathrm c}}\) and
\(\beta^{\mathrm c}|_{\Sigma^{\mathrm c}}\) are local diffeomorphisms onto open
subsets of \(P^{\mathrm c}\). Consequently,
$\Phi_{\Sigma^{\mathrm c}}
=
\beta^{\mathrm c}|_{\Sigma^{\mathrm c}}
\circ
\left(\alpha^{\mathrm c}|_{\Sigma^{\mathrm c}}\right)^{-1}$
is a local diffeomorphism of the base.

It remains to use the Legendrian condition. For contact groupoids, local
Legendrian bisections act on the base by local contact transformations
\cite{IglesiasMarrero2003,CrainicZhu2007}. This is the contact analogue of the
standard symplectic groupoid statement that Lagrangian bisections act by
Poisson transformations \cite{MarreroMartinDeDiegoStern2015}. Equivalently,
it follows by passing to the associated homogeneous symplectic
\(\Rtimes\)-principal bundle \cite{BGG2017}. Therefore
\(\Phi_{\Sigma^{\mathrm c}}\) preserves the contact structure on
\(P^{\mathrm c}\).

In a local contact form \(\theta\), preservation of the contact distribution is $T\Phi_{\Sigma^{\mathrm c}}(\ker\theta)=\ker\theta$. Equivalently, the one-forms \(\theta\) and
\(\Phi_{\Sigma^{\mathrm c}}^*\theta\) have the same kernel. Since both define
the same cooriented contact hyperplane distribution locally, there exists a
nowhere-vanishing smooth function \(f\) such that
$\Phi_{\Sigma^{\mathrm c}}^*\theta=f\,\theta$. Thus the map induced by a local Legendrian bisection is contact, or conformal
contact after choosing a local contact form.
\end{proof}

The preceding discussion should be read as the groupoid-level model for the
local construction below. We do not construct here a global contact groupoid
associated with a general skew algebroid. In fact, for a general skew
algebroid there need not be an integrating Lie groupoid. Thus, in what follows,
the role of the groupoid-level Legendrian relation is played locally by a
discrete contact Tulczyjew relation associated with a chosen admissibility
relation \(\mc A_d\subset E\times E\).

We work in a fixed local trivialization of the line bundle \(\Lbundle\to E^*\),
so that the contact variable is represented by a scalar \(z\), and we fix a
discrete admissibility relation \(\mc A_d\subset E\times E\). This relation plays the role of a local replacement for groupoid-level
composability, or for a constrained version of composability in which only
certain discrete steps are allowed. The discrete contact Tulczyjew relation constructed below
is the local representative of the Legendrian relation described above.

\begin{definition}
A discrete admissibility relation on \(E\) is a submanifold
$\mc A_d\subset E\times E$. Its elements \((a_0,a_1)\in\mc A_d\) are said to be admissible consecutive
discrete states. A discrete admissible path is a sequence
\((a_0,\dots,a_N)\) in \(E\) such that $(a_k,a_{k+1})\in\mc A_d$ for $k=0,\dots,N-1$.
\end{definition}

The relation \(\mc A_d\) should be understood as a chosen discrete replacement
of the continuous admissibility condition \(\dot x=\rho(a)\). It specifies
which pairs of algebroid elements are allowed to occur as consecutive discrete
states. Thus \(\mc A_d\) is part of the discretization data; it is not a
canonical object determined by the skew algebroid alone. In the integrable
Lie algebroid case, one expects such a relation to arise from composability,
or from a constraint relation, on an integrating groupoid. For a general skew
algebroid, where such a groupoid need not exist, \(\mc A_d\) has to be prescribed
as part of the discrete model.

Let \(\mc A_d^{\mathrm c}:=\mc A_d\times\R\), and let
\(L_d:\mc A_d^{\mathrm c}\to\R\) be a discrete contact Lagrangian, or discrete
contact generating object. The discrete Herglotz update is
\begin{equation}
z_{k+1}=z_k+L_d(a_k,a_{k+1},z_k).
\label{eq:discrete-herglotz-update}
\end{equation}
This is the same additive convention used in discrete contact mechanics \cite{SimoesMartinDeDiegoLainzDeLeon2021} for
\(L_d:Q\times Q\times\R\to\R\). The difference is that \(Q\times Q\) is
replaced by the admissibility relation \(\mc A_d\) on the skew algebroid $E$.

Let $\operatorname{pr}_0,\operatorname{pr}_1:E\times E\longrightarrow E$ denote the projections onto the first and second factors. For
\((a_0,a_1)\in E\times E\), with \(a_j\in E_{x_j}\), a tangent vector to
\(E\times E\) is vertical in the first factor if it is tangent to the fiber
\(E_{x_0}\) in the \(a_0\)-slot and has zero component in the base variables of
that slot. Similarly for the second factor. Restricting a covector on
\(E\times E\) to these vertical directions gives two maps
\[
\nu_1:T^*(E\times E)\longrightarrow E^*,
\qquad
\nu_2:T^*(E\times E)\longrightarrow E^*,
\]
covering the base points \(x_0\) and \(x_1\), respectively.

More explicitly, choose local coordinates \((x^i,y^\alpha)\) on \(E\), associated
with a local basis \((e_\alpha)\) of section of \(E\to M\) and dual basis
\((e^\alpha)\) of \(E^*\to M\). On \(E\times E\) we use coordinates $(x_0^i,y_0^\alpha,x_1^i,y_1^\alpha)$. If
\[
\alpha
=
A_i\dd x_0^i+B_\alpha\dd y_0^\alpha
+
C_i\dd x_1^i+D_\alpha\dd y_1^\alpha
\in T^*_{(a_0,a_1)}(E\times E),
\]
then the vertical restrictions are
\[
\nu_1(\alpha)=B_\alpha e^\alpha(x_0)\in E^*_{x_0},
\qquad
\nu_2(\alpha)=D_\alpha e^\alpha(x_1)\in E^*_{x_1}.
\]
Thus \(\nu_1\) extracts the fiber momentum conjugate to the first discrete
algebroid variable, while \(\nu_2\) extracts the fiber momentum conjugate to the
second one.

If \(\mathcal A_d\) is a proper submanifold of \(E \times E\), the differential \(dL_d\)
is a covector along \(\mathcal A_d\), not a covector on the ambient space \(E \times E\).
In local formulas we choose an extension \(\widetilde L_d\) of \(L_d\) to a
neighborhood of \(\mathcal A_d\) and apply \(\nu_1,\nu_2\) to \(d\widetilde L_d\).
Changing the extension changes \(d\widetilde L_d\) by an element of the conormal
bundle \(N^*\mathcal A_d\). Thus the corresponding momenta are defined modulo the usual
conormal ambiguity of constrained Tulczyjew constructions (see for instance \cite{deLeonJimenezMartinDeDiego2012}).

The expressions \(\nu_1(dL_d)\) and \(\nu_2(dL_d)\) should therefore be read as
elements defined modulo the images of \(N^*\mathcal A_d\) under the vertical projections
\(\nu_1\) and \(\nu_2\), respectively. Equivalently, the matching condition is
the vanishing of the corresponding covector on admissible interior variations.
In what follows, \(\nu_1(dL_d)\) and \(\nu_2(dL_d)\) denote any such local
representatives of the left and right momenta.

We shall use the notation
\[
\sigma_d(a_0,a_1,z):=1+\frac{\partial L_d}{\partial z}(a_0,a_1,z).
\]
This is the non-vanishing factor which appears in the discrete Herglotz
equations and in the discrete contact Legendre transformations. Assuming \(\sigma_d\neq 0\), the right and left discrete contact Legendre maps
are
\[
\FL_d^+(a_0,a_1,z):=\nu_2(\dd L_d(a_0,a_1,z)),
\qquad
\FL_d^-(a_0,a_1,z):=
-\sigma_d(a_0,a_1,z)^{-1}\nu_1(\dd L_d(a_0,a_1,z)).
\]
In local coordinates,
\[
\FL_d^+=\frac{\partial L_d}{\partial y_1},
\qquad
\FL_d^-=-\frac{1}{1+\partial_zL_d}\frac{\partial L_d}{\partial y_0}.
\]
For \(L_d\) independent of \(z\), one has \(\sigma_d=1\), and these reduce to
the usual discrete Legendre maps.

\begin{definition}The discrete contact Tulczyjew relation generated by \(L_d\) is the image of the
map
\[
\mc A_d\times\R\longrightarrow E^*\times E^*\times\R\times\R
\]
defined by
\[
(a_0,a_1,z)\longmapsto
\left(
\FL_d^-(a_0,a_1,z),
\FL_d^+(a_0,a_1,z),
z,
z+L_d(a_0,a_1,z)
\right).
\]
Equivalently, \(\mc R_{L_d}\subset E^*\times E^*\times\R\times\R\) consists of
the quadruples \((p^-,p^+,z,z')\) for which there exists
\((a_0,a_1)\in\mc A_d\), with \(a_0\in E_{x_0}\) and \(a_1\in E_{x_1}\), such
that
\[
p^-=\FL_d^-(a_0,a_1,z)\in E^*_{x_0},
\qquad
p^+=\FL_d^+(a_0,a_1,z)\in E^*_{x_1},
\qquad
z'=z+L_d(a_0,a_1,z).
\]
Thus \(\mc R_{L_d}\) is the discrete analogue of the continuous generated
relation \(\Lambda_L=\varepsilon_E^{\mathrm c}\circ J^{1,\mathrm c}L\).\end{definition}

\begin{remark} In
the groupoid interpretation, $\mc R_{L_d}$ is the local representative of the Legendrian
relation \(\Sigma^{\mathrm c}\subset\mathcal G^{\mathrm c}\). The discrete
dynamics is obtained by concatenating elements of \(\mc R_{L_d}\), namely by
matching the outgoing momentum of one step with the incoming momentum of the
next.\end{remark}

\begin{definition}
A discrete admissible path \((a_0,\dots,a_N)\) is called a discrete Herglotz
extremal if, for a fixed initial value \(z_0\), the sequence
\((z_0,\dots,z_N)\) defined recursively by
\[
z_{k+1}=z_k+L_d(a_k,a_{k+1},z_k)
\]
has stationary terminal value \(z_N\) under all variations of the intermediate
points through discrete admissible paths, with \(a_0\), \(a_N\), and \(z_0\)
fixed. \end{definition}

Note that the variations of the variables \(z_k\) are not independent; they are
induced by the discrete Herglotz recurrence.

\begin{theorem}
\label{thm:discrete-ELH}
Let \((a_0,\dots,a_N)\) be a discrete admissible path, and let
\((z_0,\dots,z_N)\) be determined by
\eqref{eq:discrete-herglotz-update}. Assume that
$\sigma_d(a_k,a_{k+1},z_k)\neq 0$ for $k=0,\dots,N-1$.
Then \((a_0,\dots,a_N)\) is a discrete Herglotz extremal if and only if, for
each \(k=1,\dots,N-1\), the covector
\[
\dd_1 L_d(a_k,a_{k+1},z_k)
+
\sigma_d(a_k,a_{k+1},z_k)
\dd_2 L_d(a_{k-1},a_k,z_{k-1})
\]
annihilates all admissible interior variations at \(a_k\). Here
\(\dd_1L_d\) and \(\dd_2L_d\) denote the components of \(\dd L_d\) with
respect to the first and second algebroid variables, respectively, understood
through a local extension when \(\mc A_d\subset E\times E\) is a proper
constraint submanifold.

Taking vertical parts, this condition gives
\[
\nu_1(\dd L_d(a_k,a_{k+1},z_k))
+
\sigma_d(a_k,a_{k+1},z_k)
\nu_2(\dd L_d(a_{k-1},a_k,z_{k-1})).
\]
Equivalently, in terms of the discrete contact Legendre maps, the corresponding
momentum matching equation is
\[
\FL_d^+(a_{k-1},a_k,z_{k-1})
=
\FL_d^-(a_k,a_{k+1},z_k),
\]
with the usual conormal interpretation when \(\mc A_d\subset E\times E\) is a
proper constraint submanifold. Thus, at the level of the discrete contact
Tulczyjew relation, discrete Herglotz extremals are represented by
concatenating elements of \(\mc R_{L_d}\).
\end{theorem}

\begin{proof}
For brevity, write \(L_k=L_d(a_k,a_{k+1},z_k)\) and
\(\sigma_k=\sigma_d(a_k,a_{k+1},z_k)\). Varying the recurrence
\(z_{k+1}=z_k+L_k\) gives
\[
\delta z_{k+1}
=
\sigma_k\delta z_k
+
\left\langle \dd_1L_d(a_k,a_{k+1},z_k),\delta a_k\right\rangle
+
\left\langle \dd_2L_d(a_k,a_{k+1},z_k),\delta a_{k+1}\right\rangle .
\]
Here \(\delta a_k\) denotes an infinitesimal variation through discrete
admissible paths. Thus, when \(\mc A_d\) is a proper submanifold of
\(E\times E\), the adjacent variations satisfy
\[
(\delta a_{k-1},\delta a_k)\in T_{(a_{k-1},a_k)}\mc A_d,
\qquad
(\delta a_k,\delta a_{k+1})\in T_{(a_k,a_{k+1})}\mc A_d .
\]

Since \(z_0\) is fixed, \(\delta z_0=0\). Iterating the previous identity yields
\[
\delta z_N
=
\sum_{j=0}^{N-1}
\left(
\prod_{m=j+1}^{N-1}\sigma_m
\right)
\left[
\left\langle \dd_1L_d(a_j,a_{j+1},z_j),\delta a_j\right\rangle
+
\left\langle \dd_2L_d(a_j,a_{j+1},z_j),\delta a_{j+1}\right\rangle
\right],
\]
where an empty product is equal to \(1\).

The endpoints \(a_0\) and \(a_N\) are fixed. Hence the covector multiplying an
interior variation \(\delta a_k\), \(1\leq k\leq N-1\), is
\[
\left(
\prod_{m=k+1}^{N-1}\sigma_m
\right)
\dd_1L_d(a_k,a_{k+1},z_k)
+
\left(
\prod_{m=k}^{N-1}\sigma_m
\right)
\dd_2L_d(a_{k-1},a_k,z_{k-1}).
\]
Stationarity of \(z_N\) is therefore equivalent to the condition that this
covector annihilates all admissible interior variations compatible with the two
adjacent constraints \((a_{k-1},a_k)\in\mc A_d\) and
\((a_k,a_{k+1})\in\mc A_d\).

Since all \(\sigma_m\) are nonzero, the common factor
\(\prod_{m=k+1}^{N-1}\sigma_m\) can be removed. Thus the stationarity condition
is
\[
\dd_1L_d(a_k,a_{k+1},z_k)
+
\sigma_k
\dd_2L_d(a_{k-1},a_k,z_{k-1})
=0
\]
on admissible interior variations. If \(\mc A_d=E\times E\), this is an
ordinary covector equality at \(a_k\). If \(\mc A_d\) is a proper constraint
submanifold, the equality is understood modulo the conormal ambiguity
associated with the local extension of \(L_d\) off \(\mc A_d\).

Taking the vertical part of this condition gives
\[
\nu_1(\dd L_d(a_k,a_{k+1},z_k))
+
\sigma_k
\nu_2(\dd L_d(a_{k-1},a_k,z_{k-1}))
=0,
\]
again with the usual conormal interpretation in the constrained case.

By definition,
\[
\FL_d^+(a_{k-1},a_k,z_{k-1})
=
\nu_2(\dd L_d(a_{k-1},a_k,z_{k-1}))
\]
and
\[
\FL_d^-(a_k,a_{k+1},z_k)
=
-\sigma_k^{-1}
\nu_1(\dd L_d(a_k,a_{k+1},z_k)).
\]
Therefore the vertical condition is represented by the momentum matching
equation
\[
\FL_d^+(a_{k-1},a_k,z_{k-1})
=
\FL_d^-(a_k,a_{k+1},z_k),
\]
again modulo conormal terms in the constrained case. This is precisely the
condition that the outgoing momentum of the step
\((a_{k-1},a_k,z_{k-1})\) coincides with the incoming momentum of the step
\((a_k,a_{k+1},z_k)\). Equivalently, consecutive elements of
\(\mc R_{L_d}\) are concatenated.
\end{proof}

\begin{remark}
The role of \(\mc A_d\) and of the conormal ambiguity is analogous to the
constrained variational construction in \cite{deLeonJimenezMartinDeDiego2012}.
There, a Lagrangian defined on a constraint submanifold generates a Lagrangian
submanifold by extending the differential modulo the conormal bundle, and the
resulting dynamics is naturally an implicit relation. The present theorem is
the contact-discrete counterpart of this mechanism. The submanifold
\(\mc A_d\subset E\times E\) replaces the continuous constraint submanifold,
the factor \(\sigma_d=1+\partial_zL_d\) accounts for the Herglotz contact
recursion, and the matching of the vertical parts of the adjacent differentials
gives the momentum matching condition. Thus the construction keeps the same
Tulczyjew--conormal philosophy, but replaces symplectic constrained
variational calculus by a contact relation generated by a discrete Herglotz
object.
\end{remark}

The preceding theorem is the local form of the contact groupoid matching
condition. If a groupoid-level Legendrian relation
\(\Sigma^{\mathrm c}\subset\mathcal G^{\mathrm c}\) exists and is represented
locally by \(L_d\), then the source and target maps of the contact groupoid are
represented by the two discrete contact Legendre maps. Thus
\(\beta^{\mathrm c}(\eta_k)=\alpha^{\mathrm c}(\eta_{k+1})\) becomes, in the
chosen trivialization,
\[
\FL_d^+(a_{k-1},a_k,z_{k-1})
=
\FL_d^-(a_k,a_{k+1},z_k),
\]
together with the update \(z_{k+1}=z_k+L_d(a_k,a_{k+1},z_k)\).

Equivalently, by grouping the variables as incoming and outgoing contact phase
points,
\[
(p^-,p^+,z,z')
\longleftrightarrow
\bigl((p^-,z),(p^+,z')\bigr),
\]
we may regard \(\mc R_{L_d}\) as a subset of $(E^*\times\R)\times(E^*\times\R)$. This is the viewpoint used below when discussing regularity.

Before discussing regularity, it is useful to regard the discrete Legendre maps
as maps to the contact phase space. Define the augmented discrete contact
Legendre maps by
\[
\mathbb F_d^-L(a_0,a_1,z)
:=
\bigl(\FL_d^-(a_0,a_1,z),z\bigr)
\in E^*\times\R,
\]
and
\[
\mathbb F_d^+L(a_0,a_1,z)
:=
\bigl(\FL_d^+(a_0,a_1,z),z+L_d(a_0,a_1,z)\bigr)
\in E^*\times\R.
\]
Thus \(\mathbb F_d^-L\) assigns to a discrete step its incoming contact phase
point, while \(\mathbb F_d^+L\) assigns its outgoing contact phase point. With
this notation, the discrete relation is the image
\[
\mc R_{L_d}
=
\bigl(\mathbb F_d^-L,\mathbb F_d^+L\bigr)(\mc A_d\times\R)
\subset
(E^*\times\R)\times(E^*\times\R).
\]

\begin{proposition}
\label{prop:discrete-regular-implicit}
Assume that the augmented discrete contact Legendre map
\(\mathbb F_d^-L\) is a local diffeomorphism onto its image. Then the discrete
relation \(\mc R_{L_d}\) is locally the graph of the map
\[
\Phi_{L_d}
=
\mathbb F_d^+L\circ
\left(\mathbb F_d^-L\right)^{-1}.
\]
Equivalently, the discrete dynamics is determined locally by the matching
condition
\[
\FL_d^+(a_{k-1},a_k,z_{k-1})
=
\FL_d^-(a_k,a_{k+1},z_k),
\qquad
z_k=z_{k-1}+L_d(a_{k-1},a_k,z_{k-1}).
\]

\end{proposition}

\begin{proof}
By definition,
\[
\mc R_{L_d}
=
\left\{
\left(
\mathbb F_d^-L(a_0,a_1,z),
\mathbb F_d^+L(a_0,a_1,z)
\right)
:
(a_0,a_1,z)\in \mc A_d\times\R
\right\}.
\]
If \(\mathbb F_d^-L\) is locally invertible onto its image, then, after
restricting to a sufficiently small neighborhood, each incoming phase point
\((p^-,z)\) in that image determines a unique discrete step
\[
(a_0,a_1,z)
=
(\mathbb F_d^-L)^{-1}(p^-,z).
\]
The outgoing phase point is therefore uniquely determined by
\[
\Phi_{L_d}(p^-,z)
=
\mathbb F_d^+L
\left(
(\mathbb F_d^-L)^{-1}(p^-,z)
\right).
\]
Consequently,
\[
\mc R_{L_d}
=
\left\{
\left(
(p^-,z),
\Phi_{L_d}(p^-,z)
\right)
:
(p^-,z)\in \operatorname{Im}(\mathbb F_d^-L)
\right\}
\]
locally. Hence \(\mc R_{L_d}\) is locally the graph of \(\Phi_{L_d}\).
This is the usual regularity mechanism behind discrete Legendre
transformations in variational integrators and groupoid discrete mechanics;
see, for instance, \cite{MarsdenWest2001,MarreroMartinDeDiegoStern2015}.

For a sequence of steps, the outgoing phase point of the step
\((a_{k-1},a_k,z_{k-1})\) is
\[
\mathbb F_d^+L(a_{k-1},a_k,z_{k-1})
=
\left(
\FL_d^+(a_{k-1},a_k,z_{k-1}),
z_k
\right),
\]
where
\[
z_k=z_{k-1}+L_d(a_{k-1},a_k,z_{k-1}).
\]
The incoming phase point of the next step \((a_k,a_{k+1},z_k)\) is
\[
\mathbb F_d^-L(a_k,a_{k+1},z_k)
=
\left(
\FL_d^-(a_k,a_{k+1},z_k),
z_k
\right).
\]
Therefore the equality of consecutive phase points is exactly the momentum
matching condition
\[
\FL_d^+(a_{k-1},a_k,z_{k-1})
=
\FL_d^-(a_k,a_{k+1},z_k),
\]
together with the discrete Herglotz update for \(z_k\).
\end{proof}

\begin{remark}
The singular case should be understood in the sense of implicit discrete
dynamics. If the augmented Legendre maps are not locally invertible, the
matching condition does not determine a discrete evolution map on
\(E^*\times\R\). Nevertheless, the subset $\mc R_{L_d}\subset (E^*\times\R)\times(E^*\times\R)$ still defines an implicit difference relation: a discrete phase sequence
\((p_k,z_k)\) is a solution whenever
\[
\bigl((p_{k-1},z_{k-1}),(p_k,z_k)\bigr)\in \mc R_{L_d}
\]
for all admissible consecutive indices. Thus the primary object is not a
map, but a relation whose points may or may not be extendable to full
forward, backward, or two-sided trajectories.

This is the same geometric viewpoint used for implicit discrete dynamics on
Lie groupoids in \cite{IglesiasMarreroMartindeDiegoPadron2013}. There, an
implicit difference equation is represented by a submanifold of a Lie groupoid,
and one may have to pass to its forward, backward, or two-sided integrable
part in order to obtain discrete trajectories. 

In the present setting
\(\mc R_{L_d}\) plays the analogous role on the contact phase space
\(E^*\times\R\). Regularity of the augmented Legendre maps is precisely the
condition which turns this implicit relation into a local graph, while failure
of regularity leaves a well-defined contact Tulczyjew relation but not, in
general, a discrete evolution map. This is the discrete analogue of the
continuous Tulczyjew picture, where singular Lagrangians define implicit
differential relations rather than vector fields.
\end{remark}

Our construction recovers standard discrete contact mechanics
\cite{SimoesMartinDeDiegoLainzDeLeon2021} when \(E=TQ\) and
\(\mc A_d=Q\times Q\). In that case the formulas above reduce to the discrete
Herglotz equations and to the left and right discrete contact Legendre
transformations of contact variational integrators. Under the usual regularity
assumptions, the induced discrete evolution is the conformal contactomorphism of
the standard theory. The present construction does not aim to reprove this
property or to reproduce the full theory of contact variational integrators,
which also includes exact discrete Lagrangians, contact exponential maps, and
backward error analysis. Its purpose is instead to isolate the Tulczyjew
relation behind the discrete Herglotz equations and to extend the
variational-contact mechanism from \(Q\times Q\) to discrete admissibility
relations on skew algebroids.

Thus, once a discrete admissibility relation has been fixed, the construction
above gives a closed local discrete contact theory. Its output is a discrete relation generated by a discrete contact object. In the regular tangent-bundle case this relation induces the standard conformal contact evolution; in the singular or skew algebroid case it remains meaningful even when no phase-space map is available. A fully global theory should replace the chosen local admissibility relation by a line-bundle, contact-groupoid, or homogeneous-symplectic groupoid construction.

This point is clarified by the contact groupoid viewpoint. In the symplectic
case, regular Lagrangian submanifolds of symplectic groupoids are local
Lagrangian bisections and generate Poisson transformations on the base. In the
contact case, the corresponding objects are Legendrian bisections of contact
groupoids, and the induced transformations are contact, or conformal contact,
maps in a local trivialization. Before imposing regularity, however, the natural
object is not a map but a Legendrian relation. The local relation
\(\mc R_{L_d}\) is the discrete algebroid representative of this relation.

\begin{remark}
Poissonization-based Jacobi Hamiltonian integrators are designed to produce
Jacobi-preserving time-step maps for Hamiltonian systems on Jacobi manifolds
\cite{AraujoOliveiraMestre2026}. The construction in this paper has a
different purpose. It starts from a discrete contact generating object and
produces a discrete contact Tulczyjew relation. When the Legendre data are
regular, this relation may induce a conformal contact map, as in the
tangent-bundle case. Without regularity, however, the relation remains defined
even when no discrete phase-space map exists. Thus the present construction is
complementary to Jacobi-preserving integrators: it is a variational--Tulczyjew
formulation of discrete contact dynamics, rather than a Poissonization-based
Hamiltonian integrator.
\end{remark}

When the skew algebroid is integrable, one expects a fully global version of
the construction. If \(G\rightrightarrows M\) integrates \(E\), the relation
\(\mc A_d\) should be replaced by the appropriate constraint or composability
relation in \(G\), and \(L_d\) should be interpreted as a contact generating
object on a contact or homogeneous-symplectic enhancement of the groupoid. The
relation \(\mc R_{L_d}\) should then arise as the local algebroid model of a
groupoid-level Legendrian relation.

Thus, once a discrete admissibility relation has been fixed, the construction
above gives a closed discrete contact theory in the local trivialization. Its
output is not merely a numerical update rule, but a discrete contact Tulczyjew
relation generated by a discrete contact object. A fully global discrete theory
should replace the trivial contactization by a line-bundle, contact-groupoid, or
homogeneous-symplectic groupoid framework, but the central mechanism is already
visible in the local construction.

\section{Conclusions and future work}\label{conc}

We have formulated a contact Tulczyjew framework for dissipative dynamics on skew
algebroids. The construction starts from the ordinary algebroid Tulczyjew morphism
and extends it in line-bundle form by adding the Euler contribution on \(E^*\). In a
local trivialization this produces exactly the contact term appearing in the local
morphism of Section~2. Thus the scalar contact variables used in coordinates are
understood as local representatives of line-bundle data.

For a contact generating object, the resulting Tulczyjew relation gives an implicit
first-order dynamics on the contact phase side. In local coordinates, the matching
condition recovers the Euler-Lagrange-Herglotz equations on the skew algebroid.
Using the admissible variations determined by the algebroid structure, these
equations are shown to be equivalent to the Herglotz variational principle. The
regular case gives a local vector field, the hyperregular case gives the
corresponding contact Hamiltonian formulation, and the singular case remains
naturally described as an implicit relation.

We have also developed a discrete contact Tulczyjew theory. The construction
starts from the contact-groupoid idea that the natural object is a Legendrian
relation, and then gives its local skew algebroid representative after a
discrete admissibility relation \(\mc A_d\subset E\times E\) has been fixed. A
discrete contact generating object determines a relation
\(\mc R_{L_d}\subset (E^*\times\R)\times(E^*\times\R)\), and the discrete
Herglotz equations are obtained by matching consecutive contact momenta, with
the usual conormal interpretation when \(\mc A_d\) is a proper constraint
submanifold. In the regular tangent-bundle case this recovers the standard
contact variational integrators, while in the singular case the same relation
remains meaningful even when no discrete phase-space map exists. Thus the
discrete theory retains the central Tulczyjew principle: dynamics is first a
geometric relation, and only under regularity assumptions does it become an
evolution map.

Several directions remain open.

\begin{enumerate}
\item \textbf{Singular dissipative systems and constraint algorithms.}
Since singular systems are naturally included in the implicit Tulczyjew picture, it is natural to develop an integrability algorithm for contact Tulczyjew relations on skew algebroids. At the level of implicit differential equations, the general geometric mechanism is to project the relation to the phase space, impose tangency, and iterate until a final integrable relation is obtained, as in the algorithm of Mendella, Marmo and Tulczyjew~\cite{MendellaMarmoTulczyjew1995}. The conservative Lie-algebroid counterpart is the constraint algorithm for presymplectic Lie algebroids and singular Lagrangian systems developed in \cite{IglesiasMarreroMartindeDiegoSosa2008}, while the contact side is modeled by the precontact approach to singular Lagrangians in \cite{deLeonLainz2019}. A contact-algebroid version should combine these ingredients: the final constraint object would no longer be a submanifold carrying a presymplectic algebroid structure, but an implicit contact Tulczyjew relation compatible with the Herglotz evolution and the algebroid admissibility conditions. Such an algorithm should clarify the role of primary and secondary constraint relations, the possible loss of uniqueness, and the passage from an implicit contact relation to a dynamics when regularity is recovered.

\item \textbf{Covariant extension to classical field theory.}
A covariant Herglotz-type theory should replace the scalar action variable by
an action current, or more intrinsically by a density-valued object. From the
present viewpoint, the goal would be to construct a covariant contact
Tulczyjew framework over first jet bundles or suitable multiphase bundles, in
which dissipative field equations arise as local expressions of intrinsic
contact generating objects. This should be compared with the
premultisymplectic Tulczyjew formulation of first-order classical field
theories, where the Euler--Lagrange and Hamilton--De Donder--Weyl equations
are described by Lagrangian submanifolds
\cite{Grabowska2012, GrabowskaGrabowski2013, CamposGuzmanMarrero2012}. The time-dependent case
\cite{GuzmanMarrero2010} provides a useful intermediate model between
mechanics and the fully covariant field-theoretic setting.

\item \textbf{Global contact groupoid formulation of the discrete theory.}
The local theory developed here is self-contained once a discrete admissibility
relation is fixed. A separate global problem is to identify geometric situations
in which this local relation is induced by a contact groupoid. So, a natural next step is to replace this
local input by a contact-groupoid or homogeneous-symplectic groupoid
construction. In the integrable Lie algebroid case, the admissibility relation
should arise from a constraint or composability relation on an integrating
groupoid, and the discrete contact Tulczyjew relation should be the local
representative of a groupoid-level Legendrian relation. This would clarify the
role of Legendrian bisections, regularity, and conformal contact transformations
in the global discrete theory.

\item \textbf{Contact optimal control on skew algebroids.}
The present framework also suggests a contact analogue of the Pontryagin
maximum principle on skew algebroids. The conservative algebroid setting is
provided by the Pontryagin maximum principle on skew algebroids
\cite{GrabowskiJozwikowski2011}, while the dissipative/contact side is modeled
by the contact Pontryagin maximum principle and the Herglotz optimal control
problem of \cite{deLeonLainzMunozLecanda2023}. From the present viewpoint, the
goal would be to formulate the necessary optimality conditions as a contact
Tulczyjew relation on \(E^*\times\R\), or intrinsically on the corresponding
line-bundle phase space. Such a theory would connect dissipative variational
dynamics, algebroid reduction, singular optimal control, and contact geometry
within a single relation-based framework.

\item \textbf{Hamilton-Jacobi theory and contact generating functions.}
Another natural direction is to develop a Hamilton--Jacobi theory for contact
Tulczyjew dynamics on skew algebroids. In the symplectic setting, the
Hamilton--Jacobi equation can be understood through generating functions and
symplectic relations, as in the classical construction described in \cite{AbrahamMarsden1978} and in the relation-based approach to generating functions.
This viewpoint is closely related to the geometric theory of the
Hamilton-Jacobi equation and generating functions developed in 
\cite{FerraroDeLeonMarreroMartinVaquero2017}. From the perspective of
the present paper, the corresponding contact problem would be to characterize
those sections, or more generally those line-bundle generating objects, whose
associated contact Tulczyjew relation is invariant under the dissipative
dynamics. Such a theory should clarify the role of Hamilton-Jacobi equations
for Herglotz systems, singular contact dynamics, and discrete contact
relations on skew algebroids.
\end{enumerate}

\end{document}